\documentclass[sigconf, nonacm]{acmart}

\usepackage[ruled,vlined]{algorithm2e}
\usepackage{balance}
\usepackage{booktabs}
\usepackage{makecell}
\usepackage{enumitem}
\usepackage{xcolor}
\usepackage{url}
\usepackage[caption=false,font=normalsize,labelfon
t=sf,textfont=sf]{subfig}

\newtheorem{thm}{Theorem}
\newtheorem{lem}[thm]{\bf Lemma}

\newtheorem{definition}{Definition}
\newtheorem{exmp}{Example}

\newcommand\vldbdoi{XX.XX/XXX.XX}
\newcommand\vldbpages{XXX-XXX}
\newcommand\vldbvolume{14}
\newcommand\vldbissue{1}
\newcommand\vldbyear{2024}
\newcommand\vldbauthors{\authors}
\newcommand\vldbtitle{\shorttitle} 
\newcommand\vldbavailabilityurl{URL_TO_YOUR_ARTIFACTS}
\newcommand\vldbpagestyle{plain} 

\begin{document}
\title{Accelerating Maximal Clique Enumeration via Graph Reduction}

\author{Wen Deng}
\affiliation{%
  \institution{Fudan University}
  \city{Shanghai}
  \country{China}
}
\email{wdeng21@m.fudan.edu.cn}

\author{Weiguo Zheng}
\affiliation{%
  \institution{Fudan University}
  \city{Shanghai}
  \country{China}
}
\email{zhengweiguo@fudan.edu.cn}

\author{Hong Cheng}
\affiliation{%
  \institution{The Chinese University of Hong Kong}
  \city{Hong Kong}
  \country{China}
}
\email{hcheng@se.cuhk.edu.hk}

\begin{abstract}
As a fundamental task in graph data management, maximal clique enumeration (MCE) has attracted extensive attention from both academic and industrial communities due to its wide range of applications. However, MCE is very challenging as the number of maximal cliques may grow exponentially with the number of vertices. The state-of-the-art methods adopt a recursive paradigm to enumerate maximal cliques exhaustively, suffering from a large amount of redundant computation. In this paper, we propose a novel reduction-based framework for MCE, namely RMCE, that aims to reduce the search space and minimize unnecessary computations. The proposed framework RMCE incorporates three kinds of powerful reduction techniques including global reduction, dynamic reduction, and maximality check reduction. Global and dynamic reduction techniques effectively reduce the size of the input graph and dynamically construct subgraphs during the recursive subtasks, respectively. The maximality check reduction minimizes the computation for ensuring maximality by utilizing neighborhood dominance between visited vertices. Extensive experiments on 18 real graphs demonstrate the effectiveness of our proposed method. It achieves remarkable speedups up to $44.7\times$ compared to existing approaches.
\end{abstract}

\maketitle

\pagestyle{\vldbpagestyle}
\begingroup\small\noindent\raggedright\textbf{PVLDB Reference Format:}\\
\vldbauthors. \vldbtitle. PVLDB, \vldbvolume(\vldbissue): \vldbpages, \vldbyear.\\
\href{https://doi.org/\vldbdoi}{doi:\vldbdoi}
\endgroup
\begingroup
\renewcommand\thefootnote{}\footnote{\noindent
This work is licensed under the Creative Commons BY-NC-ND 4.0 International License. Visit \url{https://creativecommons.org/licenses/by-nc-nd/4.0/} to view a copy of this license. For any use beyond those covered by this license, obtain permission by emailing \href{mailto:info@vldb.org}{info@vldb.org}. Copyright is held by the owner/author(s). Publication rights licensed to the VLDB Endowment. \\
\raggedright Proceedings of the VLDB Endowment, Vol. \vldbvolume, No. \vldbissue\ %
ISSN 2150-8097. \\
\href{https://doi.org/\vldbdoi}{doi:\vldbdoi} \\
}\addtocounter{footnote}{-1}\endgroup

\ifdefempty{\vldbavailabilityurl}{}{
\vspace{.3cm}
\begingroup\small\noindent\raggedright\textbf{PVLDB Artifact Availability:}\\
The source code, data, and/or other artifacts have been made available at \url{\vldbavailabilityurl}.
\endgroup
}

\section{Introduction}\label{intro}

As one of the most important cohesive structures in graph data, clique is closely related to other fundamental problems like independent set problem~\cite{tsukiyama1977new} and graph coloring problem~\cite{dukanovic2007semidefinite}. 
In an undirected graph $G$, a clique refers to a subgraph of $G$ where every pair of vertices are adjacent. A clique is maximal when no other vertices can be included to form a larger clique. Maximal clique enumeration (MCE) is the task of listing all the maximal cliques in a graph $G$ and can be applied in a variety of fields, such as computational biology~\cite{topfer2014viral,abu2005relative,yu2006predicting,matsunaga2009clique}, social network~\cite{wen2016maximal,lu2018community}, and wireless communication networks~\cite{biswas2013maximal}. 

\subsection{Existing Methods and Limitations}

\begin{figure}[tbp]
\centerline{\includegraphics[width=0.5\textwidth]{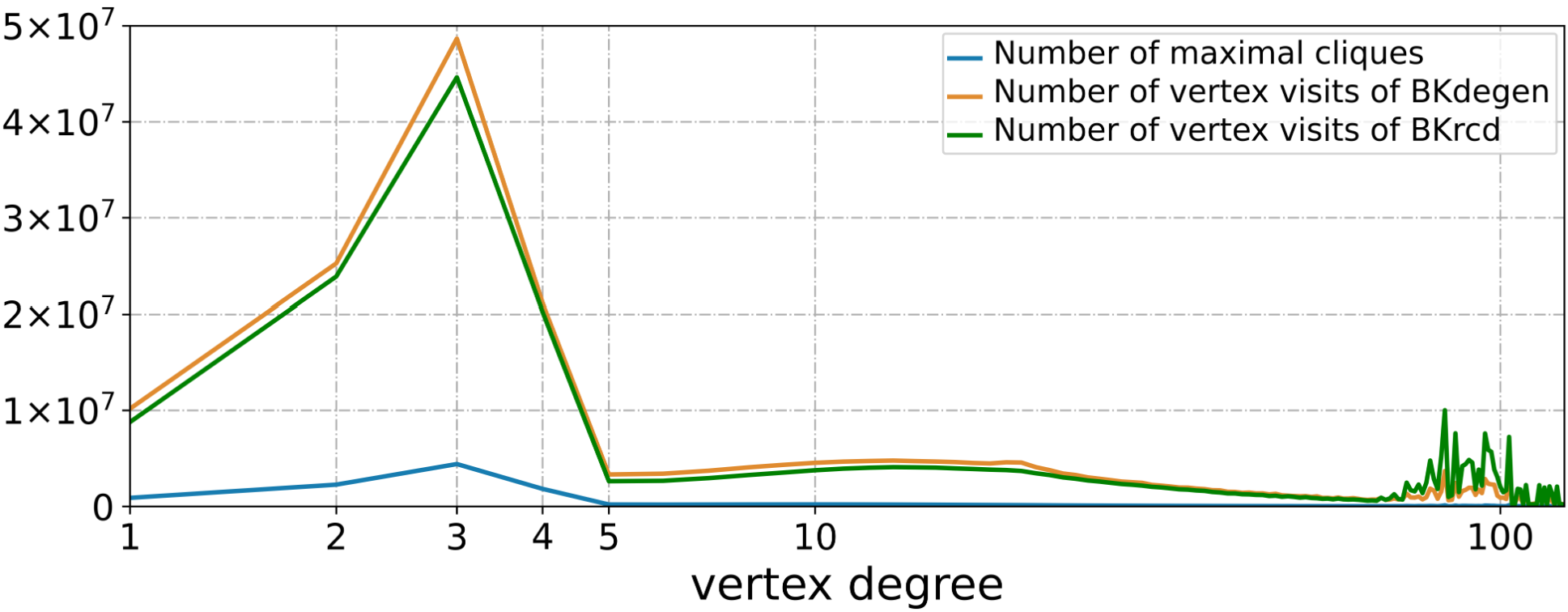}}
\vspace{-0.15in}
\caption{Illustration of the gap between the number of maximal cliques and the number of vertex visits (the horizontal axis is log-scaled).}
\label{fig_motivation}
\vspace{-0.15in}
\end{figure}

\begin{algorithm}[b]
    \caption{$BK(R,P,X$)}
    \label{alg:BKframework}
    \small
    \KwIn {Partial clique $R$, Candidate set $P$, Forbidden set $X$}
    \KwOut {All maximal cliques in $G[P]$ restricted by $X$}
    \nl \If{$P=\emptyset$ and $X=\emptyset$}{
        \nl report $R$ as a maximal clique \\
    }
    \nl \For{$v\in P$}{
        \nl $BK(R\cup \{v\}, P\cap N(v), X\cap N(v)$) \\
        \nl $P \gets P\setminus\{v\}$ \\
        \nl $X \gets X\cup \{v\}$ \\
    }
\end{algorithm}

BKdegen~\cite{eppstein2010listing} and BKrcd \cite{li2019fast} are two state-of-the-art algorithms for MCE, both of which adopt the Bron-Kerbosch (BK) 
framework~\cite{bk1973} in Algorithm \ref{alg:BKframework} that recursively enumerates all maximal cliques. The framework involves three sets, i.e., $R, P$, and $X$, where $R$ stores the partial clique, $P$ records the candidate set, and $X$ contains vertices that have already been visited to ensure the maximality (also called a forbidden set). 
The recursive function {$BK$} initializes $R,P$, and $X$ as $\emptyset, V,$ and $ \emptyset$, respectively. To expand a new branch, a vertex $v$ is moved from $P$ to $R$, Then $P$ and $X$ are updated as $P \cap N(v)$ and $X \cap N(v)$, respectively (lines 3-4). After completing this search branch, vertex $v$ is moved from $P$ to $X$ (lines 5-6). Once both $P$ and $X$ are empty, $R$ is reported as a maximal clique (lines 1-2).
BKdegen \cite{eppstein2010listing} combines the degeneracy order and pivot selection, effectively bounding the subproblem scale of each vertex by the degeneracy $\lambda$ of graph $G$. 
BKrcd \cite{li2019fast} leverages the dense nature of subproblems to enumerate maximal cliques in a top-down manner. 
A  question naturally arises ``can we make the task of maximal clique enumeration even faster?''

In practice, a substantial amount of overhead is incurred due to the need for repeated visits to specific vertices within the graph during the process of maximal clique enumeration. 
Figure \ref{fig_motivation} presents the distribution of the number of maximal cliques in which each vertex appears and the number of visits to each vertex by different algorithms. 
The results are averaged over 7 real graphs from SNAP~\cite{snapnets}.
We observe a notable gap between the number of maximal cliques and the number of vertex visits, especially for low-degree vertices.
For instance, on average, BKdegen and BKrcd visit degree-3 vertices approximately 48.6 million and 44.6 million times, respectively. However, the average number of maximal cliques involving degree-3 vertices is only about 4 million.
This large gap 
suggests that there is significant scope for further improvement in time efficiency.

Motivated by this observation, in order to enhance the performance, the key idea is to develop effective techniques to 
\textit{bridge this gap by reducing the graph size and minimizing unnecessary vertex visits during the recursive computation process, while preserving the completeness of maximal cliques.}

\subsection{Our Approach and Contributions}


\begin{figure}[tbp]
\centerline{\includegraphics[width=0.5\textwidth]{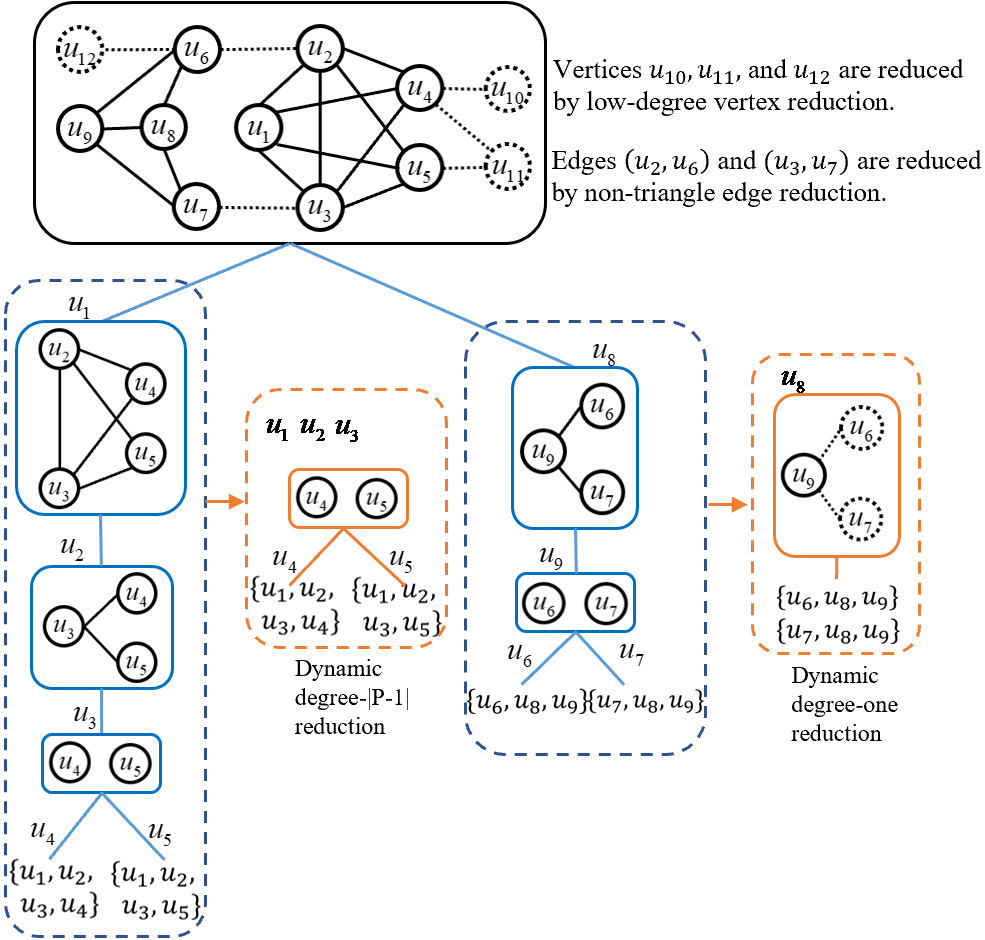}}
\vspace{-0.1in}
\caption{
The search process of a toy graph, 
where the branches in blue dashed boxes denote the original branches following the BK algorithm, while those in orange dashed boxes denote the branches by applying dynamic reduction. 
}
\label{fig_mov_search_tree}
\vspace{-0.2in}
\end{figure}

\begin{exmp}
\textit{
Let us consider the graph 
at the top of Figure~\ref{fig_mov_search_tree},
where vertex  $u_{10}$ participates in just a single maximal 2-clique with its only one neighbor $u_4$.
However, following the BK framework in Algorithm~\ref{alg:BKframework}, there are at most 5 (i.e., the number of neighbors of $u_4$) subproblems involved in the recursive computation, where the candidate set $P$ and the forbidden set $X$ may be intersected with the neighborhood of $u_4$. This will lead to unnecessary visits of vertex $u_{10}$.
%
%
By removing vertex $u_{10}$ and its corresponding edge, we can prevent their duplicate visits during the recursion, 
significantly reducing the computation cost. 
Meanwhile, we can ensure the completeness of solutions by reporting the maximal cliques that include the removed vertex $u_{10}$ in advance.
}
\end{exmp}

 Inspired by the example above, we propose a novel \underline{r}eduction-based framework for \underline{m}aximal \underline{c}lique \underline{e}numeration, namely RMCE, that incorporates three kinds of powerful reduction techniques including global reduction, dynamic reduction, and maximality check reduction.

 \noindent\textit{(1) Global Reduction}.
 Given a graph $G$, RMCE removes the low-degree vertices (whose degree is not larger than 2), because we can identify and maintain the maximal cliques involving these deleted vertices beforehand. 
 Moreover, we remove edges that do not form triangles with other edges throughout the graph, as they directly constitute maximal 2-cliques. 
 For example, in 
 Figure~\ref{fig_mov_search_tree}, 
 the edges $(u_2,u_6)$ and $(u_3,u_7)$ can be removed as they are not contained in any other cliques, forming maximal cliques themself. 
 Removing these vertices and edges can substantially reduce redundant computations during 
 the recursive procedure. 

\noindent\textit{(2) Dynamic Reduction}. Enumerating maximum cliques operates in a recursive manner, creating an extensive number of subtasks (also named subproblems). In each subtask, a subgraph will be dynamically constructed such that new degree-zero and degree-one vertices may appear. RMCE recursively reduces the subgraph size by removing these vertices.
Moreover, the subgraph is usually very dense and may contain the vertices that are adjacent to all other vertices in the candidate set $P$.
Since every maximal clique in this subtask must contain these vertices, we can move these vertices from the candidate set $P$ into the partial clique $R$ directly, and thus reduce the number of recursive calls (we call it dynamic degree-$\lvert P-1\rvert$ reduction).
 In the BK algorithm with pivot selection, each recursive call takes $O((|P|+|X|)^2)$ to choose a pivot. If $k$ recursive calls are eliminated in a subproblem, the time cost will be reduced to $O((|P|+|X|)^2)$ from $O((k+1)(|P|+|X|)^2)$. For example, in the second subgraph of the subproblem created by expanding $u_8$ in the search tree in Figure~\ref{fig_mov_search_tree}, $u_6$ and $u_7$ become new low-degree vertices that can be reduced. Removing $u_6$ and $u_7$ will result in pruning the branches expanded by $u_6$ and $u_7$. In the left branch, $u_1,u_2$, and $u_3$ are adjacent to all other vertices in the subgraph, we can move them to the partial clique $R$ together, thus decreasing the number of recursive calls from 3 to 1.

\noindent\textit{(3) Maximality Check Reduction}. 
The concept ``maximal clique'' involves two criteria: being a clique and achieving maximality. However, the existing methods 
have primarily focused on reducing the candidate vertices, with little attention given to optimizing the forbidden set $X$ that is used to perform maximality checks. The forbidden set size $\lvert X\rvert$ can be quite large, but not every vertex in $X$ is required in many subtasks, 
especially when the neighbors of a vertex (e.g., $u_i \in X$) are contained by the neighbors of another vertex (e.g., $u_j \in X$). In such cases, this vertex (e.g., $u_i$) can be removed safely from $X$ without reporting cliques that are not maximal. 
As studied in \cite{han2018speeding}, set intersections take 73.6\% of the running time in MCE. Since the forbidden set $X$ is frequently intersected during the recursion of MCE, reducing the size of $X$ can lead to a reduction in the cost of set intersection operations.
We develop an efficient algorithm with linear space overhead that minimizes the forbidden set $X$, significantly reducing unnecessary computations.


In summary, we make the following contributions in this paper.
\begin{itemize}

\item To the best of our knowledge, we are the first to propose a reduction-based framework, namely RMCE, for enumerating maximal cliques. 

\item We develop powerful global reduction techniques and dynamic reduction techniques that effectively reduce the size of the input graph globally and the size of subgraphs in recursive subtasks, respectively.

\item We introduce a novel concept of maximality check reduction for maximal clique enumeration and propose an efficient algorithm to minimize the forbidden set $X$.

\item The proposed graph reduction techniques above are orthogonal to the existing BK-based methods for maximal clique enumeration. 

\item Extensive experiments on 18 real networks demonstrate that our proposed algorithms achieve significant speedups compared to state-of-the-art approaches.

\end{itemize}



\section{Problem Definition and Preliminary}\label{sec2}

In this section, we first formally define the problem in Section~\ref{subsec:problem formulation} and then introduce two state-of-the-art methods in Section~\ref{subsec:existing solutions}. Table~\ref{notationtable} lists the frequently-used notations in the paper.

\subsection{Problem Formulation}\label{subsec:problem formulation}

Let $G = (V,E)$ be an undirected graph with $n=\lvert V \rvert$ vertices and $m=\lvert E\rvert$ edges. The neighbor vertices of a vertex $V$ in $G$ are denoted as $N(v)$, and the degree of $v$ is denoted as $d(v) = \lvert N(v) \rvert$. The common neighbors of vertices in $S$ are denoted as $C(S) = \cap_{v\in S}N(v)$.

\begin{definition}
\textit{
{\textbf{(Induced Subgraph)}}.
Given a subset $S$ of $V$, an induced subgraph $G[S]$ is defined as $G[S] = (S, E')$, where $S\subseteq V$ and $E' = \{(u, v) \in E \mid u, v \in S\}$.} 
\end{definition}

Let $N_S(v)$ denote the set of vertices in $S$ that are adjacent to vertex $v$ in graph $G$, and let $d_S(v)$ denote the number of such vertices, that is, $N_S(v) = N(v) \cap S$ and $d_S(v) = \lvert N_S(v) \rvert$. 

\begin{definition}
\textit{
{\textbf{($k$-Core)}}
Given a graph $G = (V, E)$, an induced subgraph $G[S]$ is a $k$-core if it satisfies that for every vertex $v \in S$, $d_{S}(v) \geq k$, and for every subset $S'\subset S$, the induced subgraph $G[S']$ is not a $k$-core.
}
\end{definition}

\begin{definition}
\textit{
{\textbf{(Core Number)}}
Given a graph $G = (V, E)$, the core number of a vertex $v$ is the maximum value $k$ of a $k$-core that contains $v$. The core number of a graph is the maximum core number among all its vertices.
}
\end{definition}

\begin{definition}
\textit{
{\textbf{(Degeneracy Order)}}
Given a graph $G = (V, E)$, the degeneracy of $G$, denoted by $\lambda$,  is equal to the core number of $G$. The order of vertices $\overrightarrow{V} = \{v_1, v_2, v_3, \ldots, v_n\}$ is called the degeneracy order of $G$ if vertex $v_i$ has the minimum degree in every induced subgraph $G[\{v_i, v_{i+1}, \ldots, v_n\}]$.
}
\end{definition}

The degeneracy order can be efficiently computed in linear time by iteratively removing the vertex with the smallest degree from the graph. In this order, each vertex $v_i$ has at most $\lambda$ neighbors in the induced graph $G[\{v_i, v_{i+1}, \ldots, v_n\}]$ among its later neighbors. We use $N^-(v)$ and $N^+(v)$ to denote the earlier neighbors and later neighbors of vertex $v$, respectively.

\begin{table}[!t]
\renewcommand{\arraystretch}{1.25}
\caption{Notations}
\label{notationtable}
\centering
\resizebox{\linewidth}{!}{
    \begin{tabular}{ll}
    \toprule
    Notations &  {Descriptions}\\
    \midrule
      $G=(V,E)$ & Undirected graph $G$ with vertex set $V$ and edge set $E$ \\
      $N(v)$ & Neighbors of vertex $v$  \\
      $C(S)$ & Common neighbors of vertices in $S$, i.e., $\cap_{v\in S}N(v)$ \\
      $d(v)$ & Degree of vertex $v$, i.e., $d(v)=|N(v)|$ \\
      $G[S]$ & A subgraph of $G$ induced by $S \subseteq V$ \\
      $d_{max}$ & Maximum degree of a graph \\
      $N_S(v)$ & Neighbors of $v$ in set $S$, i.e., $N_S(v)= N(v) \cap S$  \\
      $d_S(v)$ & Number of neighbors of $v$ in set $S$, i.e., $d_S(v)=|N_S(v)|$ \\
      $\lambda$ & Degeneracy of a graph \\
      $N^+(v)$ & Neighbors of $v$ whose order is larger than $v$ \\
      $N^-(v)$ & Neighbors of $v$ whose order is smaller than $v$ \\
      $R$ & Partial clique \\
      $P$ & Candidate set \\
      $X$ & Forbidden Set \\
    $(R,P,X)$ & Subproblem with $R$, $P$, and $X$ \\
      $mc(G)$ & Maximal cliques in graph $G$ \\
      $\tilde{mc}(R,P,X)$ & Maximal cliques in subproblem $(R,P,X)$ \\
      $\alpha(\Delta V,\Delta E)$ & Maximal cliques containing vertices in $\Delta V$ or edges in $\Delta E$ \\
      $\tilde{\alpha}(R,\Delta P,X)$ & \makecell[l]{Maximal cliques containing vertices in $\Delta P$ in subproblem \\ $(R,\Delta P,X)$} \\
    \bottomrule
    \end{tabular}
}
\end{table}

\begin{definition}
\textit{
{\textbf{(Clique)}}.
An induced subgraph $G[S]$ of $G=(V,E)$ is called a clique if 
there is an edge $(u,v)\in E$ for any two vertices in $S$. We denote a clique with $k$ vertices as a $k$-clique.
}
\end{definition}

\begin{definition}
\textit{{\textbf{(Maximal Clique)}}.
A clique $G[S]$ is a maximal clique in graph $G$ if and only if $\forall v \in V\setminus S,$ the subgraph induced by $S\cup \{v\}$ is not a clique.
}
\end{definition}

\begin{exmp}
    \textit{
    Let us consider the graph shown in Figure \ref{fig_mov_search_tree}. The vertices $u_1, u_2$, and $u_3$ induce a 3-clique, but it is not maximal since we can add vertex $u_4$ or $u_5$ to form a larger 4-clique with the vertices $\{u_1, u_2, u_3,u_4\}$ or 
 $\{u_1, u_2, u_3,u_5\}$. The two cliques are maximal because there are no other vertices in the graph that can be included to form a larger clique with these vertices.
    }
\end{exmp}

\noindent \textbf{Problem Statement.  \textit{(Maximal Clique Enumeration)}}. \textit{Given a graph $G=(V,E)$, the set of maximal cliques in $G$ is denoted by $mc(G)$. The task of maximal clique enumeration (shorted as MCE) is to report all the maximal cliques $mc(G)$. 
}

\subsection{Existing Solutions}\label{subsec:existing solutions}

In the next, we 
briefly review the state-of-the-art methods for the MCE problem.

\noindent\textbf{BKdegen} \cite{eppstein2010listing}: Eppstein et al. introduce the degeneracy ordering into MCE before calling the recursive function BKpivot in \cite{tomita2006worst}.
BKpivot utilizes a pivot selection strategy that selects the vertex $u$ from $X \cup P$ that has the most neighbors in set $P$, as presented in line 4 in Algorithm \ref{alg:BKdegen}.  This choice ensures that only $u$ and its non-neighbors will be involved in the subsequent search branches. The pivot mechanism divides the search space into two parts, one containing the pivot vertex $u$ and the other without it, thereby avoiding some redundant search branches.
As shown in Algorithm \ref{alg:BKdegen}, the degeneracy order is calculated at the beginning (line 1). Then, for each vertex $v$, the forbidden set $X$ is initialized as $N^-(v)$, and the candidate set $P$ is initialized as $N^+(v)$ (line 3). This ensures that every subproblem starting from $v_i$ has a candidate set $P$ whose size is no larger than the degeneracy $\lambda$. Consequently, the worst-case time complexity is reduced to $O(3^{\frac{\lambda}{3}})$.

\noindent\textbf{BKrcd} \cite{li2019fast}: 
The subgraph induced by $N^+(v)$ may be very dense due to the nature of degeneracy ordering. 
For a dense subgraph, it only needs to delete a small number of vertices to obtain a maximal clique. As outlined in Algorithm~\ref{alg:BKrcd}, BKrcd uses a top-down search strategy to remove a vertex $v$ with the fewest neighbors in $P$ (line 4) until the remaining vertices in $P$  form a maximal clique (line 3). Then, it calls BKrcd again on the removed vertex $v$ and its neighborhood (line 5). When $P$ is already a clique and passes the maximality check, $P \cup R$ is reported as a maximal clique (lines 8-9). However, not all vertices' neighborhoods 
induce a dense subgraph
in practice.

\begin{algorithm}[t]
    \caption{BKdegen($G$)}
    \label{alg:BKdegen}
    \SetKwFunction{proc}{BKpivot}
    \KwIn {Input graph $G$}
    \KwOut {All maximal cliques in $G$.}
    \nl compute the degeneracy order of the input graph $G$ \\
    \nl \For{$v \in G$}{
        \nl \proc($\{v\}, N^+(v), N^-(v)$) \\
    }
    \SetKwProg{myproc}{Procedure}{}{}
      \myproc{\proc{$R,P,X$}}{
      \nl choose a pivot $u \gets \arg\max_{v\in X\cup P}{N(v)\cap P}$ \\
      \nl \For{$w \in (P\setminus N(u)\cap P)$}{
        \nl \proc($R\cup\{w\},P\cap N(w),X\cap N(w)$)\\
        \nl $P \gets P\setminus\{w\}$ \\
        \nl $X \gets X\cup \{w\}$ \\
      }}
\end{algorithm}

\begin{algorithm}[t]
    \caption{BKrcd($R,P,X$)}
    \label{alg:BKrcd}
    \KwIn {Partial clique $R$, Candidate set $P$, Forbidden set $X$}
    \KwOut {All maximal cliques in $G[P]$ restricted by $X$}
    \nl \If{$P=\emptyset$ and $X=\emptyset$}{
        \nl report $R$ as a maximal clique \\
    }
    \nl \While{$P$ is not a clique}{
        \nl $v \gets \mathop{\arg\min}_{v\in P} \lvert N(v)\cap P\rvert$ \\
        \nl BKrcd($R\cup \{v\}, P\cap N(v), X\cap N(v)$) \\
        \nl $P \gets P\setminus\{v\}$ \\
        \nl $X \gets X\cup \{v\}$ \\
    }
    \nl \If{$P\neq \emptyset$ and $C(P)\cap X = \emptyset$}{
        \nl report $P\cup R$ as a maximal clique \\
    }
\end{algorithm}

\section{Reduction-based Framework RMCE}\label{framework}

In this subsection, we present a novel reduction-based framework for enumerating maximal cliques. 

\noindent\textbf{Reduction-based Maximal Clique Enumeration (shorted as RMCE)}. The basic idea is to reduce the graph size and recursive search space by removing vertices and edges prior to performing the exhaustive search. 

Algorithm \ref{alg:reducedMCE} depicts the details of the proposed RMCE framework.  
Initially, we apply the global reduction on the graph $G$ (line 1)  before computing its vertex order (line 2). Subsequently, for each vertex in the ascending order of the reduced graph $G$, we employ the maximality check reduction before entering the \textit{recursive} function (lines 3-6). The $recursive$ function can be any BK-based algorithm, such as BKpivot and BKrcd. The dynamic reduction is conducted at the beginning of each \textit{recursive} function (line 7).

\noindent \textbf{Superiority of RMCE}. Our reduction algorithm offers two significant advantages. Firstly, it effectively reduces the search branches during the maximal clique enumeration process, leading to a more efficient exploration of the solution space. Secondly, enormous set intersections are involved in the recursive process. RMCE can reduce the number of neighbors for vertices, enabling faster set intersections to enhance time efficiency. 

The proposed RMCE consists of three powerful reduction techniques 
including graph reduction, dynamic reduction, and maximality check reduction.

\begin{algorithm}[t]
    \SetNoFillComment
    \caption{RMCE($G$)}
    \label{alg:reducedMCE}
    \SetKwFunction{proc}{recursive}
    \KwIn {Graph $G$}
    \KwOut {All maximal cliques in $G$}
    \nl $G \gets$  apply global reduction on $G$ \\
    \nl compute an order of the reduced graph $G$ \\
    \nl \For{$i = 1 : n$}{
        \nl $v\gets$ $i$-th vertex in degeneracy order \\
        \nl $X \gets$ apply maximality check reduction on $X$ \\
        \nl \proc($R, P, X$) \\
    }
    \SetKwProg{myproc}{Procedure}{}{}
      \myproc{\proc{$R,P,X$}}{
      \nl $R,P,X \gets$ apply dynamic reduction on $(R,P,X)$ \\
      \nl choose a pivot $u$ \\
      \nl \For{$w \in (P\setminus N(u)\cap P)$}{
        \nl \proc($R\cup\{w\},P\cap N(w),X\cap N(w)$)\\
        \nl $P \gets P\setminus\{w\}$ \\
        \nl $X \gets X\cup \{w\}$ \\
      }}
\end{algorithm}

\begin{enumerate}[leftmargin=0.5cm, itemindent=0cm]
    \item \textbf{Global Reduction}: Global reduction means deleting some vertices and edges and reporting their maximal cliques in advance without breaking the completeness of the solution. Formally, if we use $mc(G)$ denote all maximal cliques of a graph $G$, then we delete some vertices $\Delta V$ and edges $\Delta E$ such that 
    $$mc(G) = mc(G') + \alpha(\Delta V, \Delta E),$$ 
    where $G'=(V\setminus \Delta V, E\setminus \Delta E)$, where $ \alpha(\Delta V, \Delta E)$ denote the maximal cliques that contain vertices in $\Delta V$ or edges in $\Delta E$. 
    This reduction technique significantly reduces the scale of the input graph, enhancing the efficiency of subsequent computations.
    \item \textbf{Dynamic Reduction}: The RMCE framework finds the maximal cliques by using a recursive search in which subproblems will be built dynamically. 
    \begin{definition}
        \textit{\textbf{(Subproblem)}. 
    In algorithms following the BK framework, a subproblem is represented by a partial clique $R$, a candidate set $P$, and a forbidden set $X$, denoted as $(R,P,X)$. The maximal cliques in this subproblem are denoted by $\tilde{mc}(R,P,X)$.
        }
    \end{definition}
    During the recursive search, the subproblem undergoes continuous changes, providing opportunities for vertices that cannot be pruned globally to be pruned within the subproblem dynamically. This brings the need for dynamic reduction.

    Formally, we delete some vertices $\Delta P_1$ from the candidate set and move some vertices $\Delta P_2$ into $R$ such that 
    {
    $$\tilde{mc}(R,P,X) = \tilde{mc}(R',P',X') + \tilde{\alpha}(R,\Delta P_1,X),$$}
    where $R'=R\cup \Delta P_2$, $P'=P\setminus (\Delta P_1\cup \Delta P_2 )$, $X'=X\cap \Delta P_2$, and $\tilde{\alpha}(R,\Delta P_1,X)$ denotes the maximal cliques that contain $R$ and vertices in $\Delta P_1$ in the subproblem $(R,P,X)$. 
    This technique further reduces the size of the subproblem, efficiently minimizing the search space required for the subsequent search.

    \item \textbf{Maximality Check Reduction}: Maximality check reduction refers to reducing the size of forbidden set $X$ without producing any cliques that are not maximal or losing any maximal cliques, that is 
    $$\tilde{mc}(R,P,X) = \tilde{mc}(R,P,X\setminus \Delta X),$$
    where $\Delta X$ denotes the deleted vertices from the forbidden set. This reduction technique prunes unnecessary computations by identifying and eliminating vertices that can be safely ignored during the maximality check process.
    \end{enumerate}

Our framework efficiently reduces both the graph size and the recursive search space. Meanwhile, the reduction itself is expected to take as little time as possible.  To achieve this, we carefully develop reduction rules to ensure that they do not introduce excessive extra computation.
Next, we will delve into the detailed design and advantages of each of these potent
reduction techniques.

\section{Global Reduction}\label{sec4}

\subsection{Intuition}

Degeneracy order, proposed by \citet{eppstein2010listing},  
has been widely employed as the preferred vertex ordering technique in the maximal clique enumeration task. Nonetheless, we argue that degeneracy order 
suffers from two major problems:
\begin{enumerate}
    \item \textbf{Redundant computations for low-degree vertices}: When generating the degeneracy order, the iterative removal of vertices with the minimum degree is necessary. In fact, identifying maximal cliques containing low-degree vertices is straightforward due to their inherent simplicity.     
    However, these vertices and their associated edges will be redundantly traversed during subsequent recursions of the MCE algorithm. 
    This redundancy 
    is expected to be mitigated through an effective reduction technique.
    \item \textbf{Redundant computations for edges}: In addition to low-degree vertices, certain edges also undergo repetitive traversals. This is particularly evident for edges that do not form triangles with other edges in high-degree vertices. The redundant computations associated with these edges can be minimized through appropriate techniques.
\end{enumerate}

We propose two kinds of reduction techniques, namely Low-degree Vertex Reduction (Section~\ref{subsec:low-degree reduction}) and Non-triangle Edge Reduction (Section~\ref{subsec:non-triangle reduction}), to effectively tackle these challenges. Our reduction methods ensure that all deleted edges and vertices are not necessary to be considered in subsequent processes, allowing the MCE algorithm to operate on the reduced graph.

\subsection{Low-Degree Vertex Reduction}\label{subsec:low-degree reduction}

In this subsection, we introduce an efficient graph reduction technique by eliminating vertices with a degree no larger than 2. 
We present three reduction rules, each handling vertices of degree zero, one, and two, respectively. 

\begin{lem}{\textbf{(Degree-Zero Reduction)}}
\textit{
A vertex $u$ with a degree of 0 can be removed from $G$ such that $mc(G)=mc(G')$ where $G' = (V\setminus \{u\}, E)$.
}
\end{lem}
\begin{proof}
    The proof is straightforward since a clique contains at least two vertices.
\end{proof}

\begin{lem}{\textbf{(Degree-One Reduction)}}
\textit{
Let $u$ be a degree-one vertex with the only neighbor $v$. The vertex $u$ and its adjacent edge can be removed such that $\lvert mc(G)\rvert=\lvert mc(G')\rvert + 1$, where $G' = (V\setminus \{u\}, E\setminus \{(u,v)\})$.
}
\label{lemma2}
\end{lem}
\begin{proof} 
The vertex set $\{u,v\}$ forms a 2-clique by definition, and $N(v)\cap N(u)=\emptyset$ ensures its maximality.
\end{proof}

\begin{figure}[tbp]
\centerline{\includegraphics[width=0.45\textwidth]{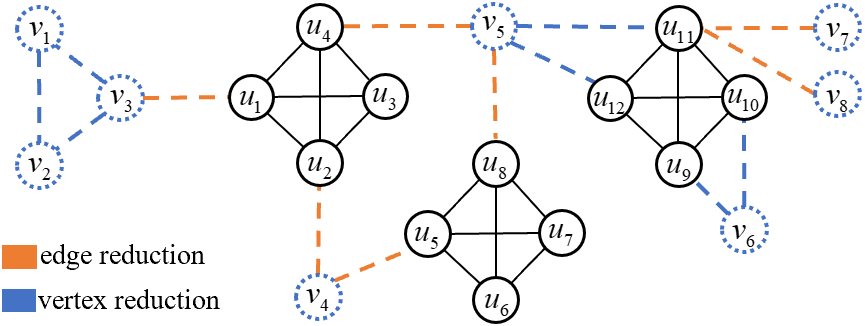}}
\vspace{-0.1in}
\caption{
Illustration of global reduction using an example graph. The vertices and edges depicted in blue dashed lines indicate those that can be deleted through our vertex reduction technique. Similarly, the edges in orange dashed lines can be removed using our edge reduction method.
}
\label{fig}
\vspace{-0.1in}
\end{figure}

\begin{lem}{\textbf{(Degree-Two Reduction)}}
\textit{For a degree-two vertex $u$ with neighbors $v$ and $w$, we have the following scenarios: 
\begin{enumerate}
    \item {If $(v,w) \notin E$}, the edges $(u,v)$ and $(u,w)$ are two maximal $2$-cliques. 
    $u$ and its edges can be deleted such that $\lvert mc(G)\rvert=\lvert mc(G')\rvert+2$, where $G' = (V\setminus \{u\}, E\setminus \{(u,v), (u,w)\})$.
    \item {If $(v,w) \in E$ and $N(v) \cap N(w) = \emptyset$}, the three edges $(u,v)$, $(u,w)$, and $(v,w)$ form a maximal $3$-clique. Vertex $u$ and the three edges can be safely deleted such that $\lvert mc(G)\rvert=\lvert mc(G')\rvert+1$, where $G' = (V\setminus \{u\}, E\setminus \{(u,v), (u,w), (v,w)\})$.
    \item {If  $(v,w) \in E$ and $N(v) \cap N(w) \neq \emptyset$}, the three edges $(u,v)$, $(u,w)$, and $(v,w)$ form a maximal $3$-clique. Vertex $u$ and two edges $(u,v)$ and $(u,w)$ can be safely deleted such that $\lvert mc(G)\rvert=\lvert mc(G')\rvert+1$, where $G' = (V\setminus \{u\}, E\setminus \{(u,v), (u,w)\})$.
\end{enumerate}
}
\end{lem}
\begin{proof}

    For the first condition, $N(u)\cap N(v)=\emptyset$ means $\{u,v\}$ forms a maximal 2-clique. Since $\nexists u'\in V, N(u')\supseteq \{u,v\}$, edge $(u,v)$ can be deleted. And edge $(u,w)$ is similar to $(u,v)$.
    When $(v,w) \in E$, the edges $(u,v)$, $(u,w)$, and $(v,w)$ form a maximal 3-clique since $N(u)\cap N(v)\cap N(w) = \emptyset$. In the second scenario, $N(v) \cap N(w) = \emptyset$, ensuring that $\{v,w\}$ cannot be enlarged, which requires the deletion of $(u,v)$ to avoid reporting it as a maximal 2-clique again. Otherwise, it implies that $\exists S\subseteq V, S\cup\{v,w\}$ is also a maximal clique, indicating that the edge $(v,w)$ cannot be deleted.
\end{proof}

\begin{exmp} \vspace{-0.1in}
\textit{
Consider the graph presented in Figure \ref{fig}. 
Vertices $v_1,v_2,v_3$, and $v_6$ can be deleted by our degree-two Reduction. The vertices $v_7$ and $v_8$ can be deleted by our degree-one Reduction. 
}
\end{exmp}

Algorithm \ref{alg:VertexReduction} provides an  overview of the \textit{VertexReduction} procedure. First, a set $Q$ is initialized with all vertices whose degree is at most 2 (line 1). Then, for each vertex $v$ in $Q$, we apply either degree-two reduction, degree-one reduction, or degree-zero reduction based on its degree (lines 4-6, 10-13, 16-17). We also update $Q$ whenever a new vertex with a degree of no larger than 2 is encountered (lines 7-10, 14-15).

\setlength{\textfloatsep}{8pt}%
\begin{algorithm}[t]
    \caption{VertexReduction($G$)}
    \label{alg:VertexReduction}
    \small
    \KwIn {Graph $G$}
    \KwOut {Reduced graph $G'$}
    \nl $G'\gets G$,  $Q \gets \{v\mid d_G(v)\leq 2\}$ \\

    \nl \For{$v \in Q$}{
        \nl $Q \gets Q\setminus\{v\}$ \\
        \nl \uIf{$d_{G'}(v) = 2$}{
            \nl $u,w \gets$  the remaining two neighbors of $v$ \\
            \nl $G' \gets$ apply the degree-two reduction rule on $v$ \\

        \nl \If{$u \notin Q$ and $d_{G'}(u) \leq 2$}{
            \nl $Q \gets Q\cup\{u\}$ \\
        }
        \nl \If{$w \notin Q$ and $d_{G'}(w) \leq 2$}{
            \nl $Q \gets Q\cup\{w\}$ \\
        }
        }
        \nl \uElseIf{$d_{G'}(v) = 1$}{
            \nl $u \gets$ the only neighbor of $v$ \\
            \nl $G' \gets$ apply the degree-one reduction rule on $v$ \\
            \nl \If{$u \notin Q$ and $d_{G'}(u) \leq 2$}{
                \nl $Q\gets Q\cup\{u\}$ \\
            }

        }
        \nl \Else{
            \nl $G' \gets$ apply the degree-zero reduction rule on $v$
        }
    }
    \nl \KwRet $G'$ \\
\end{algorithm}

\noindent\textbf{Complexity Analysis}. Algorithm \ref{alg:VertexReduction} examines a total of $n$ vertices. Checking the existence of a common neighbor in degree-two reduction has a worst-case time complexity of $O(2d_{max})$ by merge-based algorithm \cite{zheng2021accelerating},  where $d_{max}$ is the maximum degree in graph $G$. Hence, the overall time complexity is $O(nd_{max})$. 
The space cost of $O(n)$ is required to maintain the vertices that need to be removed. 

Algorithm \ref{alg:VertexReduction} demonstrates superior practical performance due to several factors. Firstly, it considers vertices with a degree less than or equal to two,
reducing the number of vertices to be processed. Secondly, for degree-zero and degree-one reductions, the algorithm operates in linear time, making the reduction efficient. Lastly, 
degree-two reduction involves finding a common neighbor between two vertices $v$ and $w$, rather than computing all the common neighbors, which runs very fast in practice.

\vspace{-0.05in}
\subsection{Non-triangle Edge Reduction}\label{subsec:non-triangle reduction}
\vspace{-0.05in}

Apart from the vertex-based reduction technique, we also propose another reduction approach from the perspective of edges, namely non-triangle edge reduction.

\begin{definition}
\textit{
    \textbf{(Non-triangle Edge)}.
    Edge $(u,v)$ is defined as a non-triangle edge if $N(v)\cap N(u) = \emptyset$. 
}
\end{definition}

A non-triangle edge $(u,v)$ means no other vertex in the graph is adjacent to both $u$ and $v$. The objective is to remove all non-triangle edges from the graph, guaranteed by the following lemma.

\begin{lem}{\textbf{(Non-triangle Edge Reduction)}}
\textit{
A non-triangle edge $(u,v)$ directly forms a maximal $2$-clique and it can be deleted from $G$ such that $\lvert mc(G)\rvert = \lvert mc(G')\rvert + 1$, where $G'=(V, E\setminus \{(u,v)\})$.  
}
\end{lem}
\begin{proof}
    It is straightforward that set $\{u,v\}$ cannot be expanded by including any other vertex, because adding any vertex into it will break the property of clique.
\end{proof}

\begin{algorithm}[t]
    \caption{EdgeReduction($G$)}
    \label{alg:EdgeReduction}
    \KwIn {Graph $G=(V,E)$}
    \KwOut {Reduced graph $G'$}
    \nl $G' \leftarrow$ G \\
    \nl \For{$(u,v) \in E$}{
        \nl \If{$(u,v)$ is not visited}{
            \nl \eIf{$u$ and $v$ has no common neighbors}{
                \nl $G' \leftarrow$ delete the edge $(u,v)$ from $G'$ \\
                \nl report $\{u,v\}$ as a maximal clique \\
            }{
                \nl $w \gets$ a common neighbor of $u$ and $v$ \\
                \nl mark edges $(u,v),(u,w)$, and $(v,w)$ visited \\
            }
        }
    }
    \nl \KwRet $G'$ \\
\end{algorithm}

\begin{figure}[tbp]
\centerline{\includegraphics[width=0.5\textwidth]{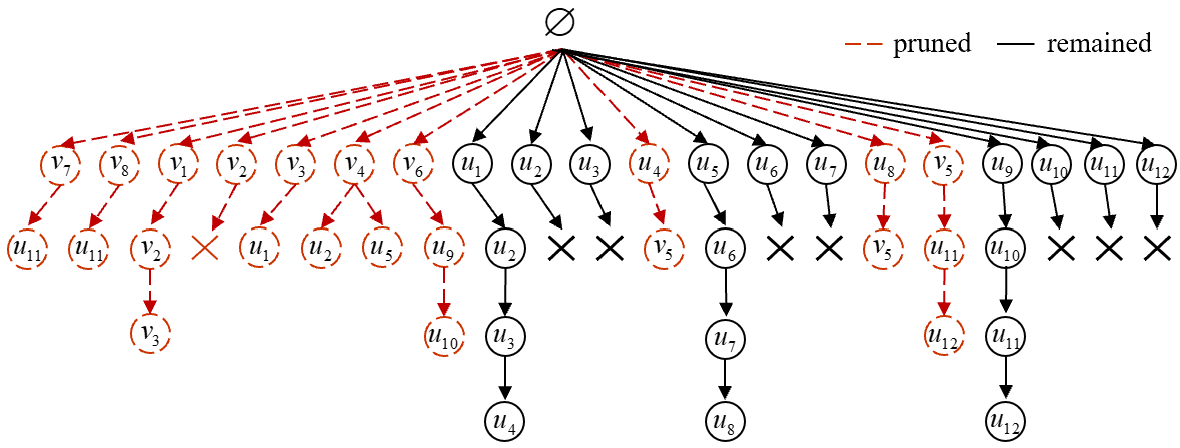}}
\vspace{-0.1in}
\caption{Search tree of BKdegen. The root $R$ is initialized to $\emptyset$.}
\label{fig2}
\vspace{-0.05in}
\end{figure}

\begin{figure*}[t]
\vspace{-0.1in}
\centering
\centerline{\includegraphics[width=1\textwidth]{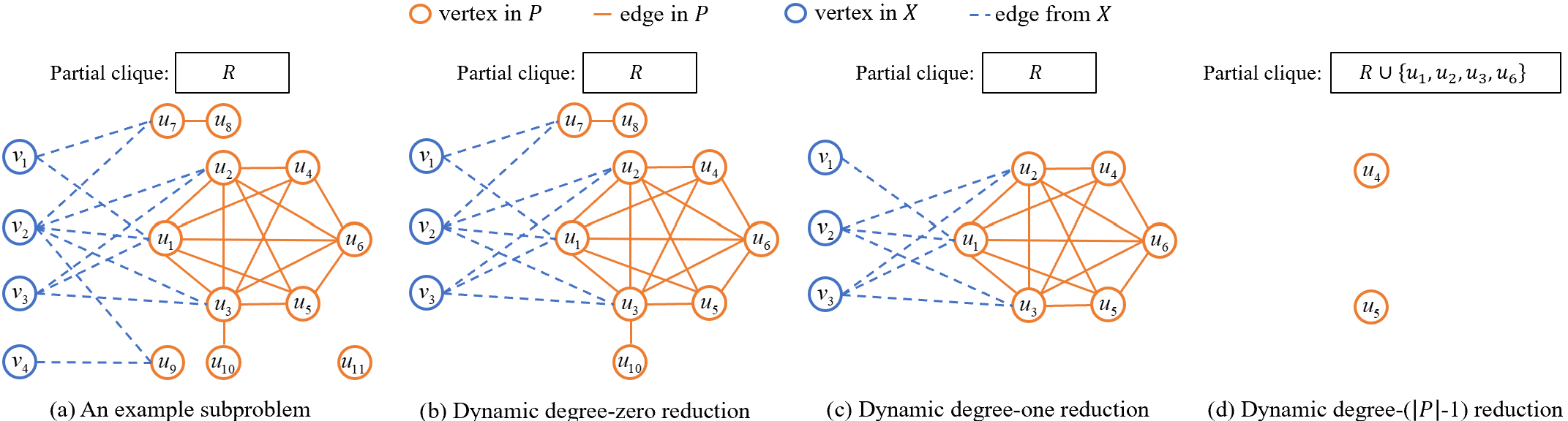}}
\vspace{-0.1in}
\caption{Illustration of dynamic reduction in the subgraph of a subproblem $(R,P,X)$.}
\label{fig3}
\vspace{-0.1in}
\end{figure*}

Algorithm \ref{alg:EdgeReduction} systematically examines each edge in the graph $G$ to determine whether it can be removed. 
Initially, for an unvisited edge $(u,v)$, if $u$ and $v$ have no common neighbors (line 4), the edge $(u,v)$ is classified as a non-triangle edge, and we can delete it while reporting $\{u,v\}$ as a maximal clique (lines 5-6). Conversely, if a vertex $w$ is their common neighbor, we mark the three edges $(u,v)$, $(u,w)$, and $(v,w)$ as visited, since they form a 3-clique (lines 7-8). Once an edge is marked as visited, it prevents the need for redundant checks on that edge in the future (line 3).

\begin{exmp} 
\textit{
Consider the graph depicted in Figure \ref{fig}. 
The edges represented by orange dashed lines (e.g., $(u_1,v_3)$ and $(v_4,u_5)$) are identified as non-triangle edges and can be safely deleted from the graph. Note that after deleting edges $(u_4,v_5)$ and $(v_5,u_8)$, vertex $v_5$ becomes a new degree-two vertex and can be removed by degree-two reduction. }

\textit{
Figure \ref{fig2} showcases a representative search tree for the graph presented in Figure \ref{fig}. By utilizing the graph reduction method, we avoid traversing the paths marked in red, resulting in a notable reduction in computational cost. Furthermore, even along the remaining black search paths, the set intersection operations 
can be accelerated due to the decreased neighborhood size for multiple vertices. 
}
\end{exmp}

\noindent\textbf{Complexity Analysis}. Algorithm \ref{alg:EdgeReduction} examines a total of $m$ edges. Moreover, finding a common neighbor takes a worst-case time complexity of $O(2d_{max})$. Consequently, the overall time complexity of \textit{EdgeReduction} is $O(md_{max})$. 
The space overhead of recording the visited edges is $O(m)$.
Notably, Algorithm \ref{alg:EdgeReduction} is considerably faster in practice due to two reasons: Firstly, once a triangle is encountered, three edges no longer require further examination. Secondly, for the majority of edges, finding a common neighbor between the two vertices only needs to identify a single vertex, resulting in a small probability of reaching worst-case time complexity.

\section{Dynamic Reduction}\label{sec5}

\subsection{Intuition}

As discussed above, global reduction directly applies to the input graph $G$. Beyond that, the recursive procedure of the MCE program will explore numerous subgraphs, presenting opportunities for further reducing the redundant computations. Let us consider the following example.

\begin{exmp}
    \textit{
   Figure \ref{fig3}(a) depicts a subgraph involved in a subproblem, where all vertices are adjacent to the current partial clique $R$. The vertices marked in blue form the forbidden set $X$, which is utilized for maximality checks, while the orange vertices represent the candidate set $P$. The corresponding recursion tree for this subproblem is displayed in Figure \ref{fig4}(a).
    By scrutinizing the recursion tree, we make the following key observations: \\
    (1) New low-degree vertices (e.g., $u_7$, $u_8$, and $u_9$) appear in this subproblem and can be effectively handled similar to the global scenario, without 
    creating additional recursion or performing set intersection operations. \\
    (2) As presented in Figure \ref{fig4}(a), there are no other branches from $u_1$ to $u_6$ in the search path $R\rightarrow u_3\rightarrow u_1\rightarrow u_2\rightarrow u_6$. Because vertices $u_1$, $u_2$, and $u_6$ connect to all other vertices in $P\cap\{u_3\}$ in the subproblem $(R\cup \{u_3\},P\cap\{u_3\},X\cap\{u_3\})$, 
    we can move $u_1$, $u_2$, and $u_6$ into partial clique $R\cup\{u_3\}$ together instead of creating three additional recursive calls. Clearly, this will further reduce the computation cost. 
    }
\end{exmp}
 
Building upon these observations, we propose a dynamic reduction technique aimed at further decreasing the size of subproblems at the recursion level.

\subsection{Dynamic Vertex Reduction}

{Dynamic reduction is designed for reducing the subgraph size of subproblem $(R,P,X)$. Different from global reduction, the maximal cliques in subgraph $G[P]$ may not be maximal in the graph $G$, because there may exist a vertex in the forbidden set $X$ that can be used to expand the maximal cliques in $G[P]$. 
To address this, we develop dynamic reduction techniques tailored specifically for degree-zero, degree-one, and degree-$(\lvert P\rvert-1)$ vertices. These techniques effectively optimize the search process while ensuring that only truly maximal cliques are reported.

For ease of presentation, a vertex $v$ is called a dynamic degree-$k$ vertex in the subproblem $(R,P,X)$ if $\lvert N_P(v)\rvert = k$.
\begin{lem}{\textbf{(Dynamic Degree-Zero Reduction)}} 
For a dynamic degree-zero vertex $u\in P$ 
in the subproblem $(R,P,X)$, we have 
\begin{enumerate} 
    \item If $N(u) \cap X = \emptyset$, $R \cup \{u\}$ is a maximal clique. Thus, we can remove $u$ from $P$, and $\lvert\tilde{mc}(R,P,X)\rvert=\lvert\tilde{mc}(R,P',X)\rvert+1$, where $P'=P\setminus \{u\}$. 
    \item If $N(u) \cap X \neq \emptyset$, $R \cup \{u\}$ is not maximal. Thus, we can remove $u$ from $P$ and $\tilde{mc}(R,P,X)=\tilde{mc}(R,P',X)$, where $P'=P\setminus \{u\}$. 
\end{enumerate}

\label{lemma4}
\end{lem}
\begin{proof}
   For the dynamic degree-zero vertex $u$, 
    we have $u \in P$, which guarantees the clique property. When $N(u) \cap X = \emptyset$, it indicates that $R\cup \{u\}$ is maximal since $N_P(u)=\emptyset$, which proves the first condition. Otherwise, adding $u$ into $R$ will produce a non-empty forbidden set $X$, making $R\cup\{u\}$ not maximal. Thus, $u$ can be removed from 
    $P$.
\end{proof}

\begin{lem}{\textbf{(Dynamic Degree-One Reduction)}} 
Consider a dynamic degree-one vertex $u\in P$ with its only neighbor $v\in P$ in the subproblem $(R,P,X)$. There are two scenarios: 
\begin{enumerate} 
    \item If $\exists w \in X$ such that $u, v \in N(w)$, vertex $u$ can be immediately removed such that $\tilde{mc}(R,P,X)=\tilde{mc}(R,P',X)$ where $P'=P\setminus \{u\}$. If $v$ is also a dynamic degree-one vertex before $u$'s removal, then $v$ should also be removed, that is $\tilde{mc}(R,P,X)=\tilde{mc}(R,P\setminus\{u,v\}, X)$. 
    \item If $\nexists w \in X$ such that $u, v \in N(w)$, $R\cup\{u,v\}$ is a maximal clique and vertex $u$ is deleted such that $\lvert\tilde{mc}(R,P,X)\rvert=\lvert\tilde{mc}(R,P',X)\rvert + 1$ where $P'=P\setminus \{u\}$. If $v$ is also a dynamic degree-one vertex before $u$'s removal, then $v$ should be removed as well, that is $\lvert\tilde{mc}(R,P,X)\rvert=\lvert\tilde{mc}(R,P\setminus \{u,v\},X)\rvert + 1$. 
\end{enumerate}
\label{lemma5}
\end{lem}
\begin{proof}
    If $\exists w \in X$ such that $u, v \in N(w)$, moving $\{u,v\}$ into $R$ would result in an empty $P$ but non-empty $X$, as $w \in X$. This indicates a non-maximal clique. Otherwise, $X$ will be empty after the intersection, signifying the discovery of a maximal clique. $v$ will become a dynamic degree-zero vertex if it is a dynamic-one vertex before. It needs to be removed because $N(u)\supseteq R\cup\{v\}$ indicates that $v$ cannot form a maximal clique in the subproblem $(R,P\setminus\{u\},X)$.
\end{proof}
}

In Lemma \ref{lemma5}, for each dynamic degree-one vertex $u \in P$ in the subproblem, we need to determine whether $u$ and $v$ share a common vertex in the forbidden set $X$. Thus, the overall time cost is $O((\lvert X\rvert + d_{max} )\lvert P\rvert)$.  Since the reduction is expected to introduce as little extra computation cost as possible, we develop a relaxed version of dynamic degree-one reduction as follows.

\begin{figure}[t]
    \centering
    \subfloat[Original Recursion Tree]{\includegraphics[width=0.22\textwidth]{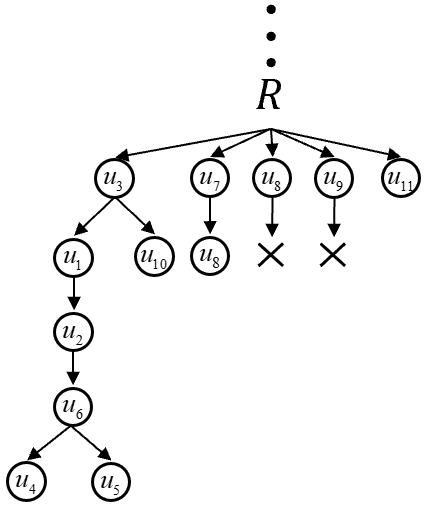}
    \label{sub_figure2.1}}
    \hfil
    \subfloat[Our Recursion Tree]{\includegraphics[width=0.22\textwidth]{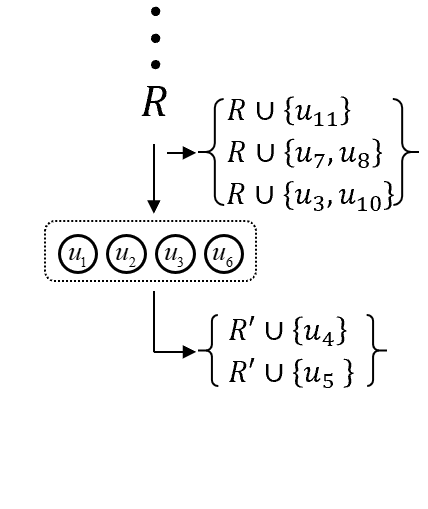}
    \label{sub_figure2.2}} 
    \vspace{-0.15in}
    \caption{Comparison of two recursion trees}
    \label{fig4}
    \vspace{-0.05in}
\end{figure}

\begin{lem}{\textbf{(Relaxed Dynamic Degree-One Reduction)}}
\textit{
Consider a dynamic degree-one vertex $u\in P$ with its only neighbor $v\in P$ in the subproblem $(R,P,X)$. 
If $N(u) \cap X = \emptyset$ or $N(v) \cap X = \emptyset$, $R \cup \{v, u\}$ forms a maximal clique and vertex $u$ can be removed from $P$ such that $\lvert\tilde{mc}(R,P,X)\rvert=\lvert\tilde{mc}(R,P',X)\rvert + 1$ where $P'=P\setminus \{u\}$. If $v$ is also a dynamic degree-one vertex before $u$'s removal, $v$ will be removed as well, that is $\lvert\tilde{mc}(R,P,X)\rvert=\lvert\tilde{mc}(R,P\setminus \{u,v\},X)\rvert + 1$. \label{lemma6}}
\end{lem}
\begin{proof}
    If $N(v) \cap X = \emptyset$ or $N(u) \cap X = \emptyset$, it implies that $\nexists w \in X$ such that $v, u \in N(w)$. Therefore, we can deduce the 
    conclusion 
    using Lemma \ref{lemma5}.
\end{proof}

\begin{algorithm}[t]
    \SetNoFillComment
    \caption{dynamicVertexReduction($R,P,X$)}
    \label{alg:dynamicVertexReduction}
    \KwIn {Partial clique $R$, Candidate set $P$, Forbidden set $X$}
    \KwOut {Reduced $R', P', X'$}
    \nl $R', P', X' \leftarrow R, P, X$ \\
    \nl \For{$v \in X$}{
        \nl \For{$u \in N_P(v)$}{
            \nl mark $u$  \\
        }
    }
    \nl \For{$v \in P$}{
        \nl \If{$d_P(v) = 0$}{
            \nl $P' \gets$ apply the dynamic degree-zero reduction rule on $v$ \\
        }
        \nl \ElseIf{$d_P(v) = 1$}{
            \nl $u \gets$ $v$'s only neighbor \\
            \nl \If{$v$ is not marked or $u$ is not marked}{
                \nl $P' \gets$ apply the relaxed dynamic degree-one reduction rule on $v$ \\
            }
        }
    }
    \nl \For{$v \in P'$}{
        \nl \If{$d_{P'}(v)=\lvert P' \rvert -1$}{
            \nl $P', R' \gets$ apply the dynamic degree-$(\lvert P\rvert-1)$ reduction rule on $v$ \\
        }
    }
    \tcc{Update the forbidden set $X$}
    \nl $X' \gets X\cap N(R')$ \\

  \nl \KwRet $R',P',X'$\\
\end{algorithm}

Note that the cost of conducting this reduction rule is trivial since we can efficiently maintain $N(v) \cap X$ by traversing the neighbors of each vertex in the forbidden set $X$ just once. Subsequently, applying Lemma \ref{lemma6} on dynamic degree-one vertices only requires traversing the candidate set $P$. The overall time complexity for this process is $O(\lvert X\rvert\lvert P\rvert)$.

\begin{exmp}
  Let us consider the subproblem presented in Figure~\ref{fig3}(a).
The vertices in the candidate set $P$ are in orange color and the forbidden set $X$ consists of vertices in blue color. Through the reduction rules, we can safely 
remove the dynamic degree-zero vertices $u_9$ and $u_{11}$, resulting in the new subgraph of the subproblem shown in Figure \ref{fig3}(b). Additionally, removing the degree-one vertices $u_7$, $u_8$, and $u_{10}$ leads to the subgraph of this subproblem in Figure \ref{fig3}(c). 
\end{exmp}

\begin{lem}{\textbf{(Dynamic Degree-$(\lvert P\rvert-1)$ Reduction)}}
\textit{If a vertex $u\in P$ satisfies $\lvert N_P(u)\rvert = \lvert P\rvert -1$ in the subproblem $(R,P,X)$, $u$ is called a dynamic degree-$(\lvert P\rvert-1)$ vertex. In this case, we can directly move $u$ from $P$ into $R$, ensuring that $\tilde{mc}(R,P,X) = \tilde{mc}(R',P',X')$, where $R'=R\cup\{u\}$, $P'=P\setminus \{u\}$, and $X'=X\cap N(u)$.
\label{lemma-n-1}}
\end{lem}
\begin{proof}
    For a dynamic degree-$(\lvert P\rvert-1)$ vertex $u$, we have $\lvert N_P(u)\rvert = \lvert P\rvert -1$, which means $\forall S\subseteq P$ and $u\notin S$, $u\in C(S)$. Assume there exists a subset $S\subseteq P$ that $u\notin S$ and $S\cup R$ is a clique. Then $S\cup R$ is not maximal due to $u \in P\cap C(S)$. That is any maximal clique in the subproblem $(R,P,X)$ must contain the vertex $u$.

\end{proof}

\begin{exmp} \label{exmp7}
Let us revisit the subproblem presented in Figure~\ref{fig3}(a). 
Based on the updated subgraph in Figure \ref{fig3}(c), we identify that vertices $u_1$, $u_2$, $u_3$, and $u_6$ are all dynamic degree-$(\lvert P\rvert-1)$ vertices. As a consequence, adding these vertices to $R$ will make the subgraph of this subproblem much smaller, with only two isolated vertices, as shown in Figure \ref{fig3}(d). Hence, the overall search tree of this subproblem derived from our reduction rule is shown in Figure \ref{fig4}(b).
It is noteworthy that the dynamic degree-zero and degree-one reductions effectively eliminate low-degree vertices from the original subproblem, thereby enhancing the efficiency of dynamic degree-$(\lvert P\rvert-1)$ reduction and further reducing the scale of the subproblem.
\end{exmp}

Algorithm \ref{alg:dynamicVertexReduction} outlines our dynamic reduction procedure in detail. It starts by marking each neighbor $u$ of each vertex $v$ in forbidden set $X$ if $u\in P$, where a marked vertex $u$ means $N(u)\cap X \neq \emptyset$ (lines 2-4). Subsequently, we iterate through each vertex $v \in P$. 
For dynamic degree-zero vertices, we remove them according to Lemma \ref{lemma4} (lines 6-7). For dynamic degree-one vertices, we selectively remove some of them following the relaxed dynamic degree-one reduction rule in Lemma \ref{lemma6} (lines 8-11). Afterward, we reiterate through reduced $P'$ again to perform the dynamic degree-$(\lvert P\rvert-1)$ reduction (lines 12-14). 
Once all reduction operations are completed, we update the forbidden set $X$ so that $R'\subseteq N(X')$ 
(line 15).

\noindent\textbf{Complexity Analysis}. 
Marking the vertices in $P$ takes the cost $O(\lambda\lvert X \rvert)$  (lines 2-4), where $\lambda$ is the degeneracy of the  graph $G$.  
Traversing the vertices in $P$ and their neighbors costs $O(\lambda\lvert P\rvert)$. Updating the set $X$ also requires $O(\lambda\lvert X \rvert)$ time. Thus, 
the overall time cost of Algorithm \ref{alg:dynamicVertexReduction} is $O(\lambda(\lvert X\rvert+\lvert P\rvert))$ in the worst case.
The space overhead is $O(\lvert P \rvert)$ as we need to mark the vertices in $P$.

\noindent \textbf{Discussion.} 
Let us consider the search trees in Figure~\ref{fig4}, there are 5 search nodes that need to conduct a pivot selection (excluding leaf nodes) in Figure~\ref{fig4}(a). The time overhead is $O(5(|P|+|X|)^2)$. By performing dynamic reduction techniques (explained in Example~\ref{exmp7}), we reduce the number of nodes that need to conduct a pivot selection from 5 to 1, as shown in Figure~\ref{fig4}(b). 
In other words, the speedups of this procedure achieve 5$\times$.
In practice, such situations may occur 
frequently 
during the search process, 
improving the time efficiency significantly.

\section{Maximality Check Reduction}\label{sec6}

\subsection{Intuition}

In order to enhance the time efficiency, 
considerable effort has been focused on reducing the candidate set $P$, while the cost of the maximality check has often been overlooked. 
In this section, we shed light on the observation that the forbidden set $X$ entails a significant amount of redundant computation. To address this issue, we introduce a technique namely, \textit{maximality check reduction}, that efficiently reduces 
the forbidden set $X$. 

\subsection{Forbidden Set Reduction}

The insight above leads to a method of reducing the forbidden set $X$ by checking the neighborhood dominance between vertices in $X$, as presented in the following Lemma.  

\begin{lem}{\textbf{(Forbidden Set Reduction by Neighbor Containment)}}
\textit{
For two vertices $u,v \in X$, if $N_P(u) \subseteq N_P(v)$, it holds that $\tilde{mc}(R,P,X) = \tilde{mc}(R,P,X')$, where $X' = X\setminus\{u\}$.
\label{lemma7}}
\end{lem}
\begin{proof}
    In the subproblem $(R,P,X)$, assume a non-maximal clique $R\cup\Delta R$ becomes maximal due to the removal of $u$ where $\Delta R \subseteq P$, which means $\Delta R \subseteq N_P(u)$. The maximality $X\cap C(\Delta R)=\emptyset$ implies that $\Delta R \supseteq N_P(v)$, which contradicts the condition $N_P(u) \subseteq N_P(v)$. 
    Thus, the Lemma \ref{lemma7} holds.
\end{proof}

\begin{exmp}
    \textit{
    Consider the graph shown in Figure \ref{fig3}(a). We observe that $N_P(v_1) = \{u_1, u_7\}$, $N_P(v_3) = \{u_1,u_2,u_3\}$, and $N_P(v_4) = \{u_9\}$ are subsets of $N_P(v_2) = \{u_1,u_2,u_3, u_7,u_9\}$. If we remove vertices $v_1$, $v_3$, and $v_4$ from the forbidden set $X$, the solution of the current subproblem will remain unaffected, that is, any cliques that are not maximal will not be reported and any maximal cliques will not be lost. For example, $R\cup\{u_9\}$ is not a maximal clique due to the adjacent relationship with vertices $v_4$ and $v_2$. After removing $v_4$ from the forbidden set, the clique $R\cup\{u_9\}$ will not be maximal due to the existence of vertex $v_2$.
    }
\end{exmp}

Establishing the neighbor containment relationship between vertices in $X$ by traversing them and their neighbors in each subtask incurs significant time overhead. 
Therefore, we design an efficient approach that continuously establishes the neighbor containment relationship between vertices during the MCE process, allowing us to prune the $X$ set without incurring additional time overhead. The overall process is outlined in Algorithm \ref{alg:forbiddenSetReduction}.

\begin{algorithm}[t]
    \SetNoFillComment
    \caption{forbiddenSetReduction($v, P, X, ignoreId$)}
    \label{alg:forbiddenSetReduction}
    \KwIn {Vertex $v$ inducing the subproblem 
    with partial clique $\{v\}$, candidate set $P$, and forbidden set $X$;
    an array $ignoreId$ 
    }
    \KwOut {Reduced $X'$}
    \nl $X'\leftarrow X$ \\
    \tcc{reduce $X$ for current subgraph}
    \nl $i \gets$ the order of vertex $v$ in ascending order \\
    \nl \For{$u \in X$}{
        \nl \If{$ignoreId[u] < i$}{
            \nl $X' \leftarrow X'\setminus\{u\}$ \\
        }
    }
    \tcc{update $ignoreId$ 
    }
    \nl \For{$u \in P$}{
        \nl \If{$P \subseteq N^+(u)$}{
            \nl $j \gets$ the order of $u$ \\
            \nl $ignoreId[v] \gets \min(j,ignoreId[v])$ \\
        }
        \nl \ElseIf{$N^+(u)\subseteq P$}{
            \nl $ignoreId[u] \gets \min(i,ignoreId[u])$ \\
        }
    }
  \nl \KwRet $X'$ \\
\end{algorithm}

A subproblem induced by vertex $v$ refers to the subproblem with partial clique $R=\{v\}$, candidate set $P=N^+(v)$, and forbidden set $X=N^-(v)$. Algorithm \ref{alg:forbiddenSetReduction} aims to establish a neighbor containment relationship between $v$ and vertices in the candidate set $P$ of the induced subproblem by vertex $v$, and efficiently prunes the $X$ set. It utilizes an array of size $n$, called $ignoreId$, to record the order in which vertices can be excluded from constructing the set $X$. Initially, all elements are set to $n$, indicating that no vertex will be ignored.
For a subgraph induced by the later neighbors of a vertex $v$, if a candidate vertex $u$ satisfies $P \subseteq N^+(u)$, i.e., $N^+(v) \subseteq N^+(u)$, then $v$ can be ignored in all subsequent subproblems after completing the $j$-th iteration, where $j$ is the order of $u$. This is because in any subproblem after that, if $v \in X$, then $u$ must be contained in $X$, which allows us to prune $v$ following Lemma~\ref{lemma7}. Thus, the $ignoreId$ of vertex $v$ will be updated as the minimum of $j$ and its current value (lines 7-9). Conversely, if $N^+(u) \subseteq P$, it means $u$ can be ignored instantly after this iteration since $N_P(u)$ will be dominated by $N_P(v)$ in all subsequent subproblems (lines 10-11). 

\noindent \textbf{Complexity Analysis}.  Reducing the size of the forbidden set $X$ runs in $O(\lvert X\rvert)$, which is linear in the size of $X$. The process of updating  $ignoreId$ for each vertex in $P$ involves traversing their later neighbors. Thus, the time complexity of Algorithm \ref{alg:forbiddenSetReduction} is $O(\lambda\lvert P\rvert)$.
The space complexity 
is $O(n)$ due to the array \textit{ignoreId} of size $n$.

\section{Experiments} \label{sec7}

\subsection{Experimental Setting}

\begin{figure*}[tbp]\vspace{-0.1in}
\centerline{\includegraphics[width=1\textwidth]{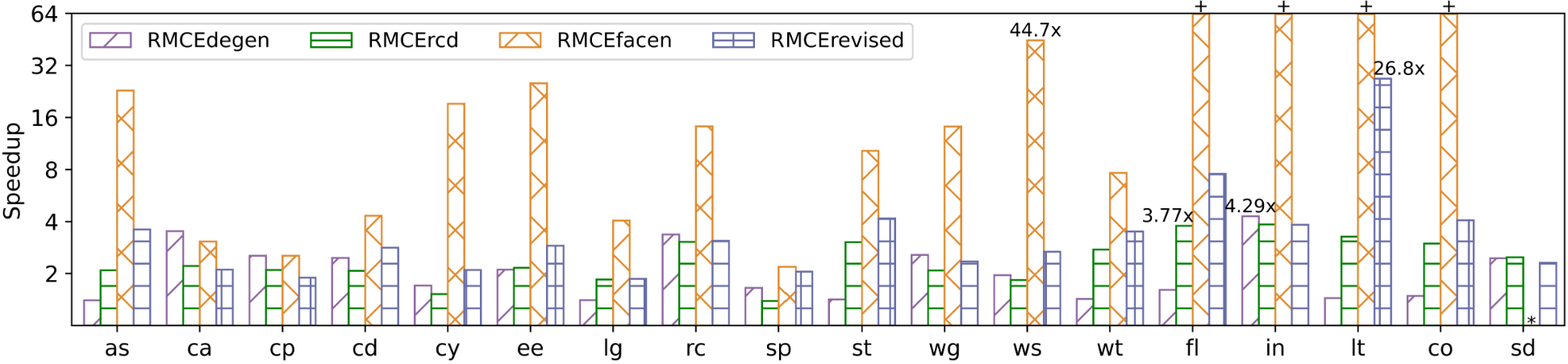}}
\vspace{-0.1in}
\caption{Overall performance in enumerating maximal cliques:  each bar represents the speedups of its corresponding algorithm enhanced by our reduction methods over the original algorithm (e.g., RMCEdegen represents $\frac{\text{running time of BKdegen}}{\text{running time of RMCEdegen}}$).
The ``+'' symbol denotes that our method completes within 12 hours while the original one fails. 
The ``*'' symbol denotes that both our method and the original algorithm failed to complete within 12 hours.}
\label{speedup}
\end{figure*}

\begin{table}[!t]
\renewcommand{\arraystretch}{1.25}
\caption{Graph Statistics, where $\lambda$ represents the degeneracy.}
\label{table1}
\vspace{-0.1in}
\centering
\resizebox{\linewidth}{!}{
\begin{tabular}{ccrrrr}
   \toprule
   Graph & Abbr. & \#Vertices &  \#Edges & $d_{max}$ & $\lambda$ \\
   \midrule
as-skitter       & as & 1696415    & 11095298  & 35455  & 111        \\ 
ca-CondMat       & ca & 23133      & 93439     & 279    & 25         \\ 
cit-Patents      & cp & 3774768    & 16518947  & 793    & 64         \\ 
com-dblp         & cd & 317080     & 1049866   & 343    & 113        \\ 
com-orkut        & co & 3072441    & 117185083 & 33313  & 253        \\ 
com-youtube      & cy & 1134890    & 2987624   & 28754  & 51         \\ 
email-EuAll      & ee & 265009     & 364481    & 7636   & 37         \\ 
flickr           & fl & 105938     & 2316948   & 5425   & 573        \\ 
inf-road-usa     & in & 23947346   & 28854311  & 9      & 3          \\ 
large\_twitch    & lt & 168114     & 6797557   & 35279  & 149        \\ 
loc-gowalla      & lg & 196591     & 950327    & 14730  & 51         \\ 
roadNet-CA       & rc & 1965206    & 2766607   & 12     & 3          \\ 
sc-delaunay\_n23 & sd & 8388608    & 25165784  & 28     & 4          \\ 
soc-pokec        & sp & 1632803    & 22301964  & 14854  & 47         \\ 
soc-twitter-higgs & st & 456631     & 12508440  & 51386  & 125        \\ 
web-Google       & wg & 875713     & 4322051   & 6332   & 44         \\ 
web-Stanford     & ws & 281903     & 1992636   & 38625  & 71         \\ 
wiki-Talk        & wt & 2394385    & 4659565   & 100029 & 131        \\ 
   \bottomrule
\end{tabular}
}
\end{table}

\noindent\textbf{Dataset}. As listed in Table~\ref{table1}, we use 18 real networks from SNAP~\cite{snapnets} and Network Repository \cite{rossi2015network} in the experiments. 

\noindent\textbf{Algorithms}. 
We evaluate the performance of enhancing four existing methods BKdegen, BKrcd, BKfacen, and BKrevised.
\begin{itemize}[leftmargin=0.35cm, itemindent=0cm]
    \item {BKdegen} \cite{eppstein2010listing}: The degeneracy-based algorithm.
    \item {BKrcd} \cite{li2019fast}: The top-down algorithm for MCE.
    \item {BKfacen} \cite{jin2022fast}: MCE algorithm that uses hybrid data structure with adjacency list and partial adjacency matrix.
    \item {BKrevised} \cite{naude2016refined}: MCE algorithm with a revised pivot selection strategy.
    \item {RMCEdegen}: 
     Our method uses the recursion of BKdegen.
    \item {RMCErcd}: 
     Our method uses the recursion of BKrcd.
     \item {RMCEfacen}: 
     Our method uses the recursion of BKfacen.
     \item  {RMCErevised}:
     Our method uses the recursion of BKrevised.
\end{itemize}

All experiments are conducted on a Linux Server equipped with
Intel(R) Xeon(R) CPU E5-2696 v4 @ 2.20GHz and 128G RAM. All
algorithms are implemented in C++ and compiled with -$O3$ option. The source code of RMCE is available at \url{https://github.com/DengWen0425/RMCE}.

\subsection{Overall Result}
Taking the running time cost of each original MCE algorithm as the baseline, we compute the speedups of our methods over BKdegen, BKrcd, BKfacen, and BKrevised, respectively. Figure \ref{speedup} presents the results for 18 graphs. 
we can observe that RMCEdegen and RMCErcd consistently outperform both BKdegen and BKrcd. Specifically, RMCEdegen achieves a maximum speedup of 4.29$\times$ in the inf-road-usa dataset,  RMCErcd achieves a maximum speedup of 3.77$\times$  in the flickr dataset,  RMCEfacen achieves a maximum speedup of 44.7$\times$ in the web-standford dataset, and RMCErevised achieves a maximum speedup of 26.8$\times$ in the large\_twitch dataset. 
The experimental results confirm that our proposed reduction methods demonstrate considerable performance in accelerating the MCE procedure.

\subsection{Detailed Evaluation}
\noindent\textit{{\underline{The Effect of Global Reduction}}}. To evaluate the effect of our global reduction, we show the ratio of deleted vertices and the ratio of deleted edges compared to the original graph in Figure \ref{reductionratio}.
We can find that over 35\% vertices in 12 graphs (e.g., as-Skitter, cit-Patents, and com-youtube) can be removed by using the global reduction. Beyond that, over 20\% edges in 9 graphs (e.g., email-EuAll, loc-gowalla, and wiki-Talk) can be removed. In particular, all vertices and edges in graphs inf-road-usa  and roadNet-CA have been deleted due to their sparsity (thus, we can exclude them from subsequent experiments). It is worth noting that no vertices and edges in the sc-delaunay\_n23 dataset have been deleted, which also verifies the effectiveness of our other reduction methods, i.e., dynamic reduction and maximality check
reduction, on the other hand.

\noindent\textit{{\underline{Number of Recursive Calls}}}. 
In this experiment, we investigate the number of recursive calls during the MCE algorithm to demonstrate the pruning power of our reduction method. We evaluate the reduction ability of three algorithms by comparing their recursive calls to that of the baseline BKdegen. 
Formally, we define the ratio of recursive calls as 
$$\frac{\#\text{recursive calls of 
a method (like RMCEdegen) 
} }{\#\text{recursive calls of original algorithm (like BKdegen)}}$$
As depicted in Figure \ref{recursiveCount}, RMCEdegen can reduce the number of recursive calls to no more than 17.6\%, RMCErcd can reduce the number of recursive calls to no more than 28.5\%, the number of recursive calls of RMCEfacen are all below 4.5\% of original, and for RMCErevised, the ratios of the number of recursive calls are no more than 20.5\%. 
In specific instances, such as inf-road-usa and roadNet-CA, 
no recursive calls are required since all the vertices and edges have been removed by utilizing the global reduction. 
The ratio of the number of recursive calls serves as an indicator of the effectiveness of a method, as each recursive call typically involves set intersection and pivot selection. 
The superior pruning ability of our algorithms 
is evident from these results.

\begin{figure}[tbp]
\centerline{\includegraphics[width=0.5\textwidth]{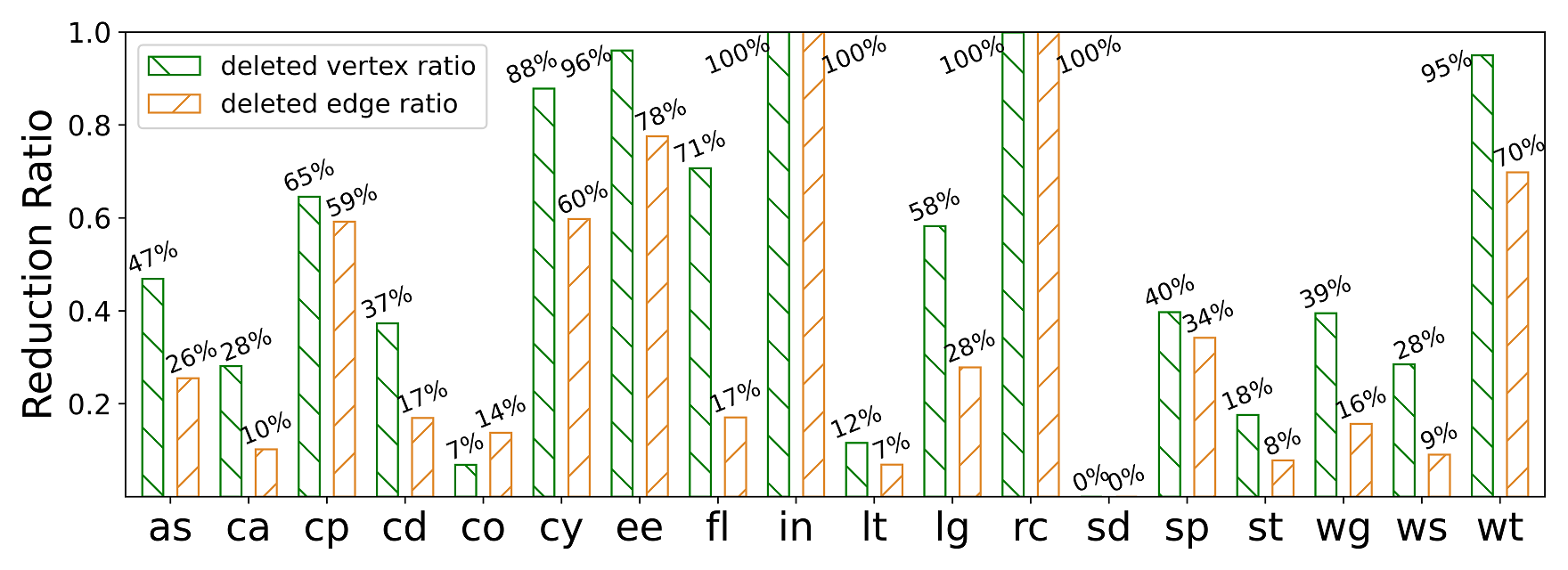}}
\vspace{-0.1in}
\caption{Reduction ratio of global reduction}
\label{reductionratio}
\vspace{-0.1in}
\end{figure}

\begin{figure}[tbp]
\centerline{\includegraphics[width=0.5\textwidth]{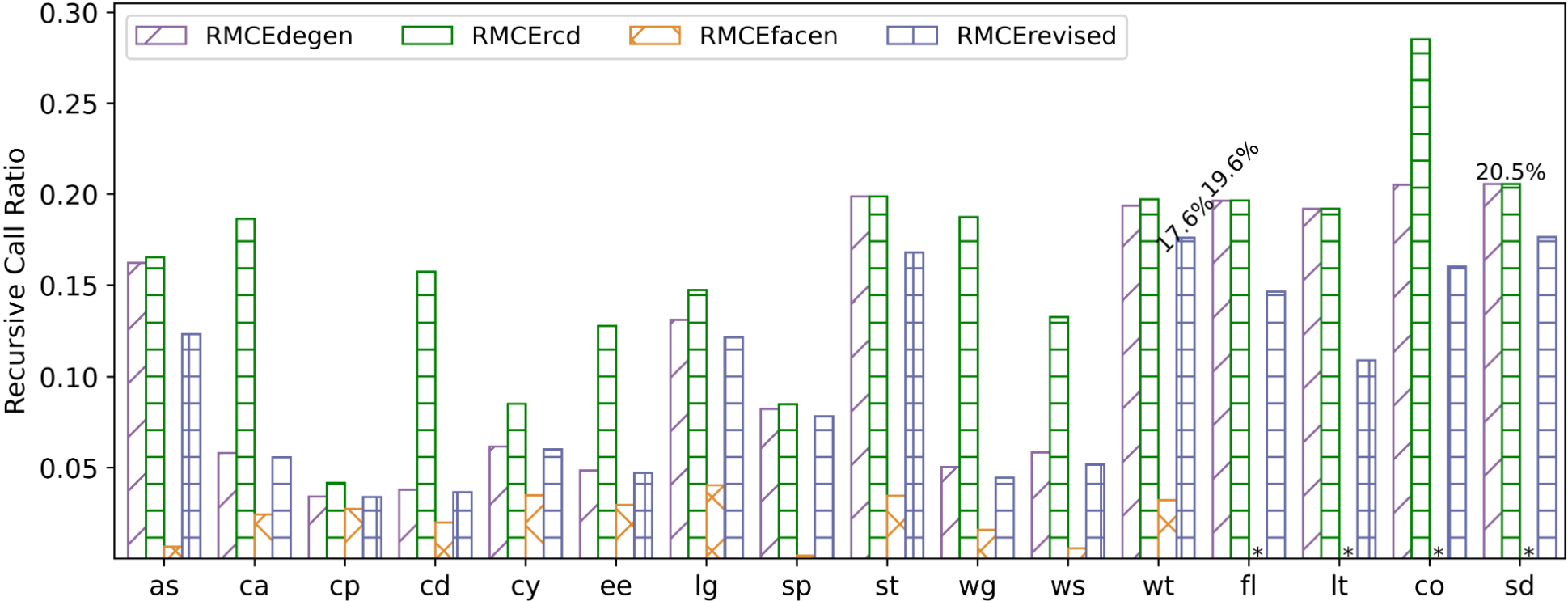}}
\vspace{-0.125in}
\caption{Ratio of recursive calls. The symbol ``*'' denotes that the original method fails to complete. 
}
\label{recursiveCount}
\vspace{-0.075in}
\end{figure}

\noindent\textit{{\underline{Effect of Forbidden Set Reduction}}}. To evaluate the impact of our forbidden set reduction, we examine two key metrics: the ratio of pruned vertices when constructing the forbidden set $X$ for each subproblem ($r_{vertex} = \frac{\sum\lvert X'\rvert}{\sum\lvert X\rvert}$) and the ratio of subproblems where forbidden set reduction occurs  ($r_{subproblem}=\frac{\# \{subproblem\mid X'\subset X\}}{\# subproblem}$) in the outer iteration. The results are depicted in Figure \ref{xprune}. 

Remarkably, we observe that in datasets such as ca-CondMat, com-dblp, web-Google, and web-Stanford, the ratio of pruned vertices (i.e., $r_{vertex}$) can reach close to 50\%. Additionally, in datasets like ca-CondMat, com-dblp, flickr, and sc-delaunay\_n23, the ratio of reduced subproblems ($r_{subproblem}$) achieves nearly 40\%, indicating the significant pruning effect of our maximality check reduction technique on the size of the forbidden set $X$. These results highlight the effectiveness of our method in reducing the computational overhead of the maximality check process.

\noindent\textit{{\underline{Reduction of Vertex Visits}}}. We also report the distribution of vertex visits with respect to their degrees. Due to space limitations, we report the results of 4 graphs including web-Google, cit-Patents, soc-pokec, and com-dblp.  
As illustrated in Figure \ref{neweval}, our method RMCEdegen yields a substantial reduction in the number of vertex visits across different degrees. 
In web-Google ((Figure \ref{neweval}(a)), RMCEdegen reduces 88\% vertex visits compared to BKdegen (70\% compared to BKrcd) at degree 20. In cit-Patents (Figure \ref{neweval}(b)), RMCEdegen reduces 82\% vertex visits compared to BKdegen (82\% compared to BKrcd) at degree 15. And in com-dblp (Figure \ref{neweval}(d)), RMCEdegen reduces 73\% vertex visits compared to BKdegen (61\% compared to BKrcd) at degree 3. These results confirm the effectiveness of global reduction.

\begin{figure}[tbp]
\centerline{\includegraphics[width=0.5\textwidth]{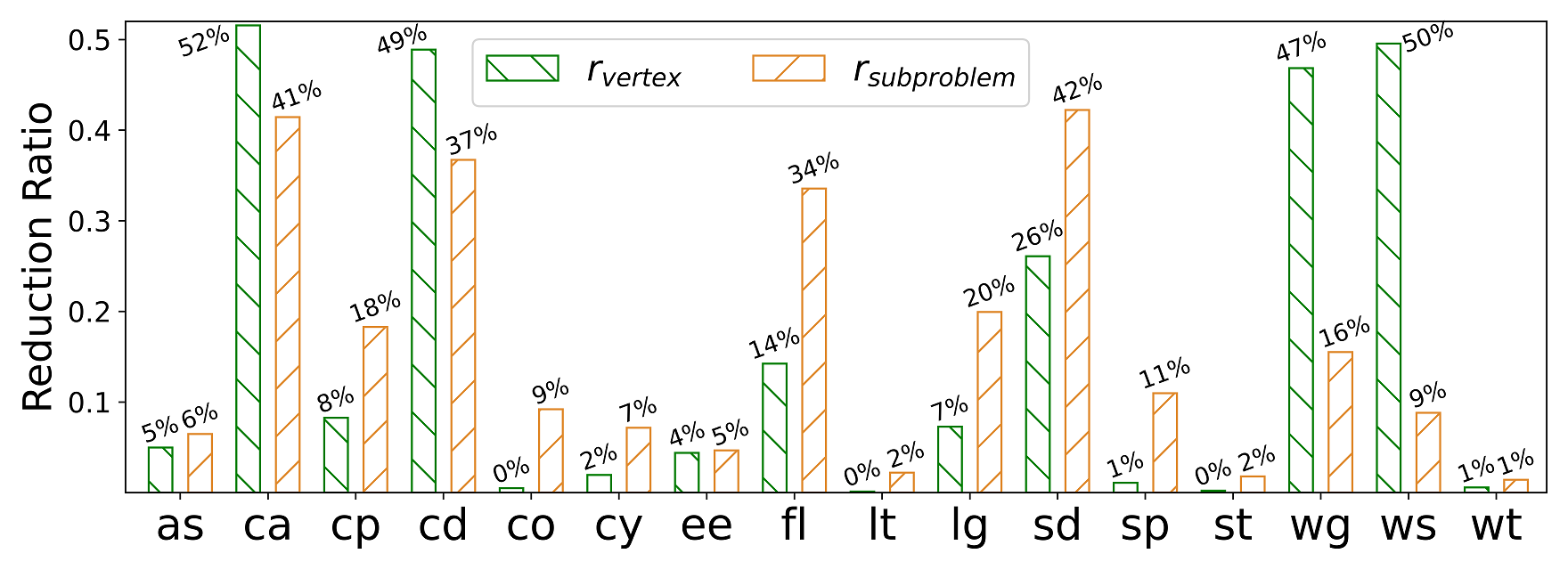}}
\vspace{-0.1in}
\caption{Reduction ratio of maximality check reduction}
\label{xprune}
\vspace{-0.1in}
\end{figure}

While for graphs that contain many vertices with a higher degree such as soc-pokec in Figure \ref{neweval}(c), our method not only decreases the visits of low-degree vertices but also reduces the visits of high-degree vertices. For example, RMCEdegen reduces 46\% vertex visits compared to BKdegen (54\% compared to BKrcd) at degree 100 in soc-pokec. Similar results can be seen 
in cit-Patents as presented in Figure \ref{neweval}(b).
These results verify the effectiveness of our dynamic reduction and maximality check reduction.

\begin{table}[!t]\vspace{-0.1in}
\renewcommand{\arraystretch}{1.25}
\caption{Ablation Study}
\vspace{-0.1in}
\label{table3}
\resizebox{\linewidth}{!}{
\centering
\begin{tabular}{crrrr}
\toprule
Graph &  RMCEdegen &  Variant1 &  Variant2 &  Variant3\\
\midrule
as-skitter        & 57.49    & \textbf{51.22}     & 70.52     & 60.77   \\ 
ca-CondMat        & \textbf{0.05}   & 0.05    & 0.06    & 0.11  \\ 
cit-Patents       & \textbf{22.14}     & 25.71     & 25.85     & 24.86   \\ 
com-dblp          & \textbf{0.67}   & 0.75     & 0.90    & 0.90  \\ 
com-orkut         & \textbf{2393.59}    & 2475.37     & 2867.58     & 2451.96   \\ 
com-youtube       & 4.01    & \textbf{3.74}     & 4.47     & 4.19   \\ 
email-EuAll       & 0.47    & \textbf{0.39}     & 0.48    & 0.44  \\ 
flickr            & \textbf{178.86}    & 184.36     & 249.78     & 185.40   \\ 
inf-road-usa      & \textbf{11.51}    & 19.07     & 11.82     & 11.62   \\ 
large\_twitch     & \textbf{325.24}    & 341.99     & 408.66      & 344.67    \\ 
loc-gowalla       & 1.91    & \textbf{1.74}     & 2.38      & 2.06   \\ 
roadNet-CA        & \textbf{0.95}   & 1.41     & 0.97    & 0.96   \\ 
sc-delaunay\_n23  & 11.52    & \textbf{9.28}      & 13.53     & 12.04   \\ 
soc-pokec         & 44.77    & \textbf{43.69}     & 49.62     & 48.93   \\ 
soc-twitter-higgs & \textbf{391.48}    & 405.62     & 478.73     & 415.12   \\ 
web-Google        & \textbf{2.55}    & 2.57     & 3.00     & 2.69   \\ 
web-Stanford      & \textbf{1.51}    & 1.52     & 2.08     & 1.53   \\ 
wiki-Talk         & 76.68    & \textbf{75.63}     & 90.74     & 80.63   \\ 
\bottomrule
\end{tabular}
\vspace{-0.1in}
}
\end{table}

\begin{figure*}[ttbp]
\vspace{-0.15in}
\centering
\subfloat[web-Google]{\includegraphics[width=0.47\textwidth]{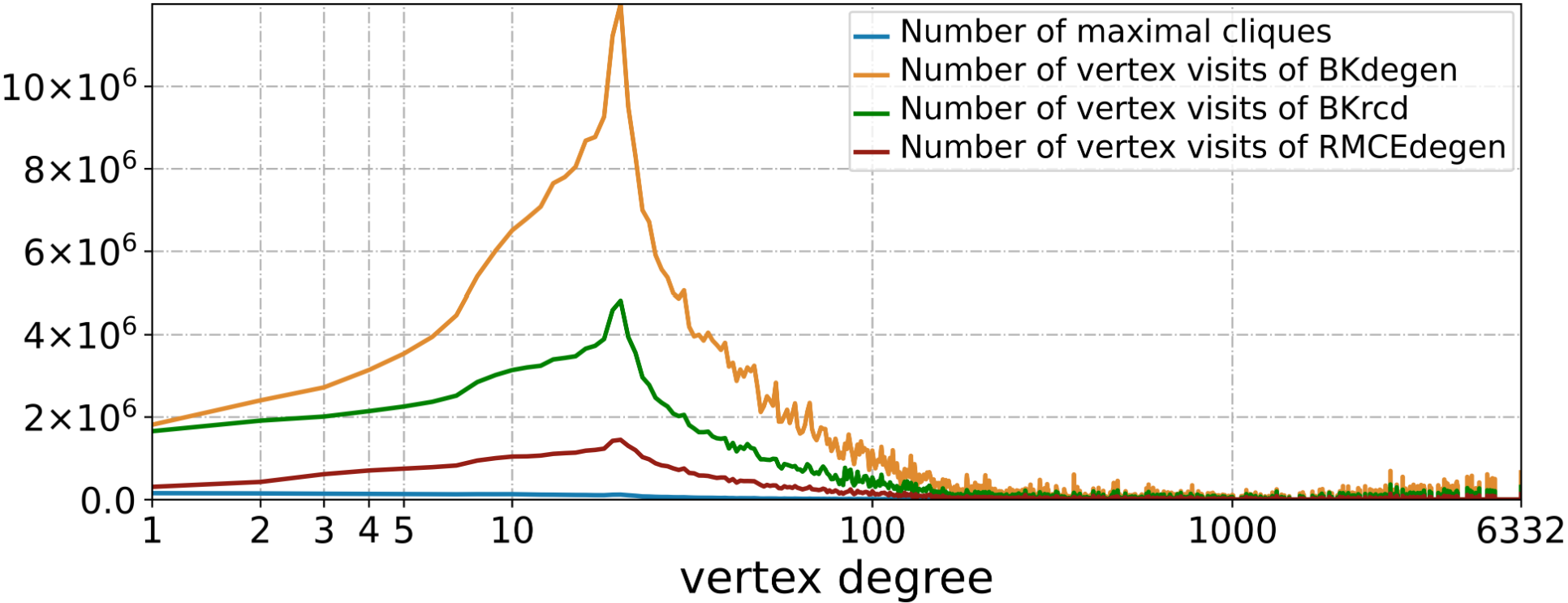}
\label{eval2.1}}
\hfil
\subfloat[cit-Patents]{\includegraphics[width=0.47\textwidth]{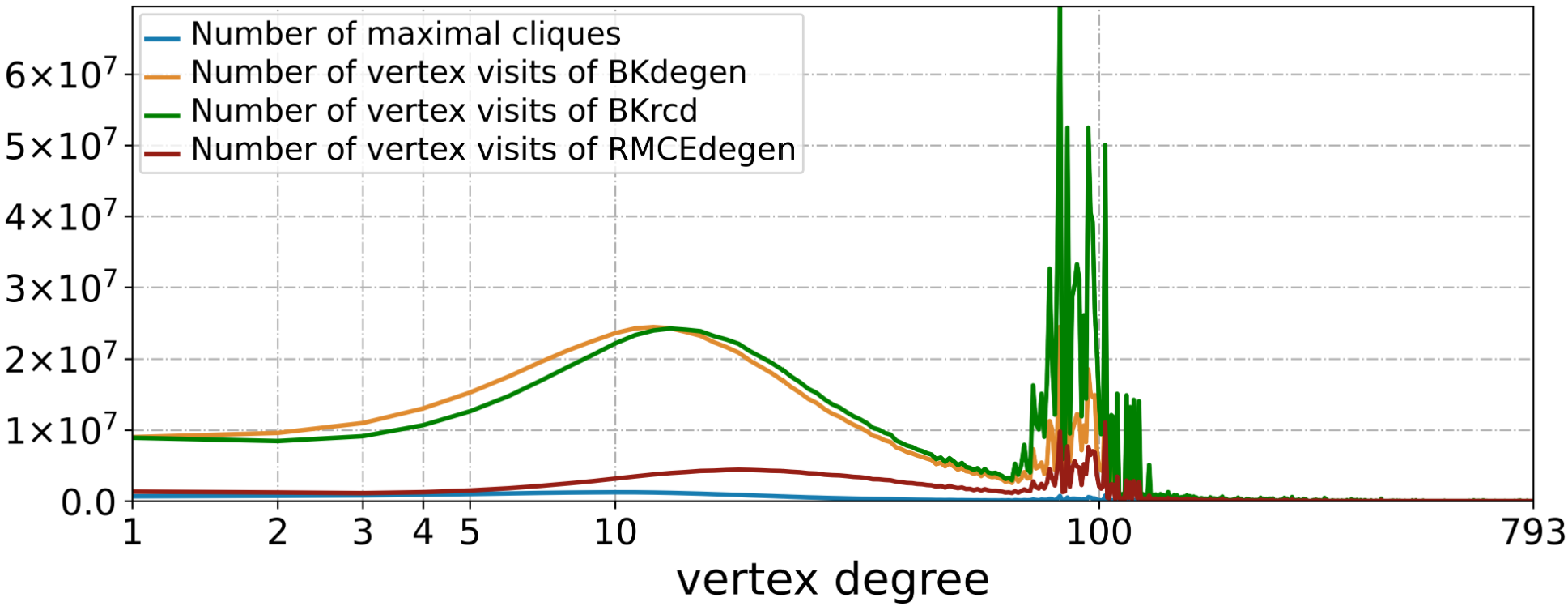}
\label{eval2.2}} \\
\vspace{-0.1in}
\subfloat[soc-pokec]{\includegraphics[width=0.47\textwidth]{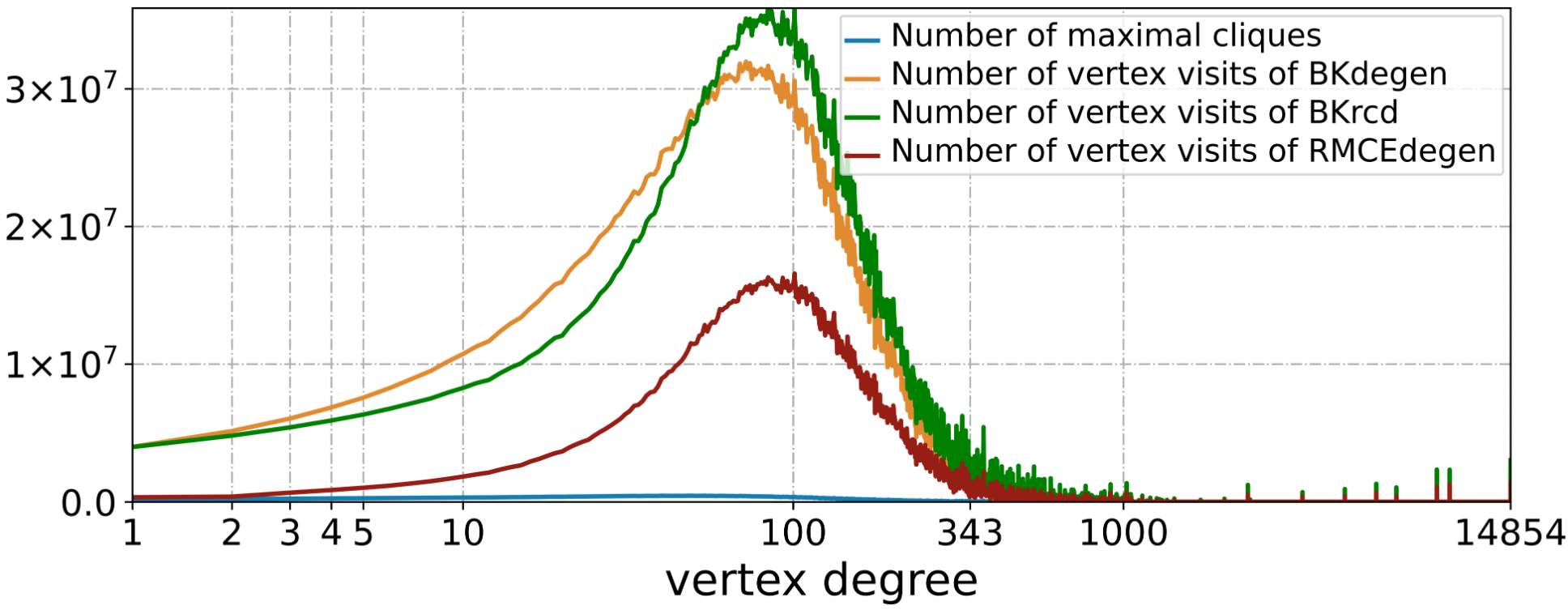}
\label{eval2.3}}
\hfil
\subfloat[com-dblp]{\includegraphics[width=0.47\textwidth]{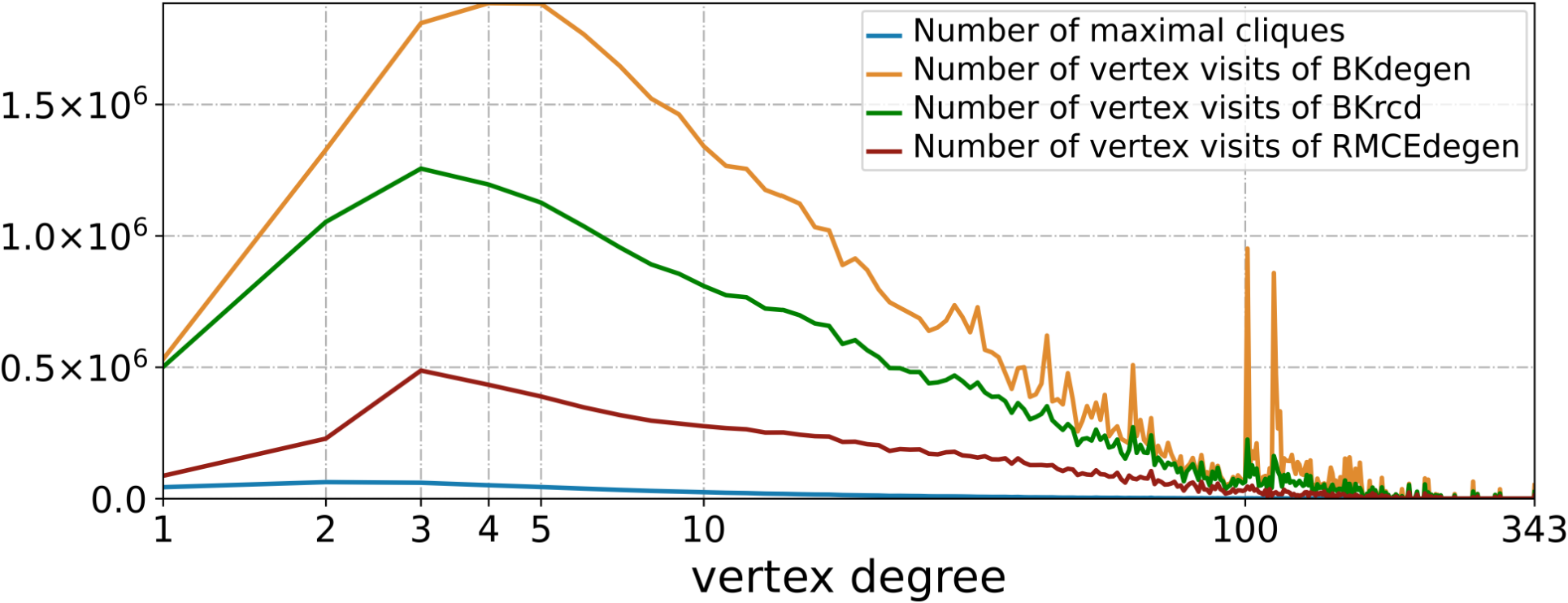}
\label{eval2.4}}
\vspace{-0.1in}
\caption{Illustration of the 
gaps between the number of maximal cliques and the number of vertex visits by different methods (the horizontal axis is log-scaled).}
\label{neweval}
\vspace{-0.1in}
\end{figure*}

\subsection{Ablation Study}

To evaluate the individual effectiveness of the proposed three reduction methods, we implement three variant algorithms, named \textit{Variant1}, \textit{Variant2}, and \textit{Variant3}, each of which disables global reduction, dynamic reduction, and maximality check reduction, respectively. Table \ref{table3} presents the running time of these variants. 

We can observe that the complete version RMCEdegen outperforms all other variants in 11 datasets. However, Variant1 runs faster in 7 datasets. This may be attributed to the overhead of global edge reduction. Nevertheless, the running time gap between Variant1 and RMCEdegen remains relatively small, confirming the overall effectiveness of global reduction.
Removing the dynamic reduction technique significantly degrades the performance in most datasets (e.g., as-skitter and flickr) since it effectively prunes a substantial number of search branches during recursion. The last column of the table demonstrates that maximality check reduction consistently improves the time efficiency across all datasets. As this technique just 
takes linear space overhead 
without incurring additional computations, it provides consistent speed-up benefits.
In summary, these results validate the efficiency and effectiveness of our reduction methods in optimizing the performance of the MCE algorithm. 

\section{Related Work} \label{sec8}

\noindent\textbf{Maximal Clique Enumeration}. The classic Bron-Kerbosch (B-K) algorithm\cite{bk1973} is a recursive backtracking algorithm that solves MCE by maintaining three sets of vertices, i.e., $R,P$, and $X$. 
They also first propose the naive pivot technique, that is to choose the first vertex as the pivot. Tomita et al. \cite{tomita2006worst} prove that the worst-case time complexity is $O(3^{\frac{n}{3}})$ by choosing the vertex $u$ which maximizes $\lvert N(u) \cap P \rvert$ from $P \cup X$ as a pivot. Eppstein et al. \cite{eppstein2010listing} further improve the B-K algorithm with degeneracy vertex order in which each vertex has at most $\lambda$ neighbors following it. This order ensures that the size of each sub-problem is at most $d$ and also reduces the worst-case time complexity to $O(3^{\frac{\lambda}{3}})$. Li et al.~\cite{li2019fast} propose a top-to-down approach, that repeatedly chooses and removes the vertex with the smallest degree until a clique is reached, which aims to efficiently solve the MCE in these dense degeneracy neighborhoods. Other studies have explored solving MCE using external memory~\cite{cheng2011finding} or using GPUs to accelerate the B-K algorithm~\cite{wei2021accelerating}. Additionally, the output-sensitive algorithm~\cite{tsukiyama1977new,makino2004new,chang2013fast,conte2016sublinear} is a branching algorithm that guarantees the time interval between two consecutive outputs (also known as delay) at the polynomial level. The enumeration time of this algorithm is related to the number of maximum cliques of the final output, making it an output-sensitive algorithm.

\noindent\textbf{Parallel Approach}. Numerous parallel algorithms have been designed for MCE \cite{du2006parallel,schmidt2009scalable,lessley2017maximal,san2018efficiently, das2018shared}. Du et al. \cite{du2006parallel} implement a method that regards each vertex and its neighborhood $N(v)$ as a basic task, with each processor responsible for multiple basic tasks according to a simple serial number mapping. Schmidt et al. \cite{schmidt2009scalable} improve load balancing of parallel algorithms through a work stealing strategy, where an idle thread randomly polls one or more search tree nodes at the bottom of the stack of other threads. Lessley et al. \cite{lessley2017maximal} introduce an approach consisting of data-parallel operations on shared-memory, multi-core architectures, aiming to achieve efficient and portable performance across different architectures. Das et al. \cite{das2018shared} also present a shared-memory parallel method that parallelizes both pivot selection and sub-problem expansion.

\noindent\textbf{MCE on Special Graphs}. 
Many works focus on solving the MCE problem in other variant graphs, such as uncertain graphs\cite{mukherjee2015mining,li2019improved,dai2022fast}, dynamic graphs\cite{sun2017mining,das2019incremental}, temporal graphs\cite{himmel2016enumerating,molter2021isolation}, heterogeneous graphs\cite{hu2019discovering}, and attributed graphs\cite{pan2022fairness,zhang2023fairness}, among others.
Beyond that, the maximum clique problem is closely related to MCE. Lijun Chang \cite{chang2019eff} introduces several efficient reduction rules to tackle the maximum clique problem. Another related task is the maximum independent set, where reduction rules have also been applied \cite{fomin2009measure,dahlum2016accelerating,chang2017computing}.

\vspace{-0.1in}
\section{conclusion} \label{sec9}

In this paper, we introduce a novel reduction-based framework RMCE for enumerating maximal cliques. To reduce the computation cost, RMCE incorporates powerful graph reductions including global reduction, dynamic reduction, and maximality check reduction. We conduct comprehensive experiments over 18 real graphs. The empirical results confirm 
the 
effectiveness of our proposed reduction techniques.



\balance

\bibliographystyle{ACM-Reference-Format}
\bibliography{sample}

\end{document}